\def\papertitle{Untangling Phase and Time in Monophonic Sounds}
\def\paperauthorA{Henning Thielemann}

\documentclass[twoside,a4paper]{article}
\usepackage{dafx-custom}
\usepackage{amsmath,amssymb,amsfonts,amsthm}
\usepackage{subfigure,color}
\usepackage{euscript}
\usepackage[latin1]{inputenc}
\usepackage[T1]{fontenc}
\ninept

\usepackage{times}

\usepackage{ifpdf}

\ifpdf 
  \usepackage[pdftex,
    pdftitle={\papertitle},
    pdfauthor={\paperauthorA},
    colorlinks=true,
    bookmarksnumbered, 
    pdfstartview=XYZ 
  ]{hyperref}
  \usepackage[pdftex]{graphicx}
  \usepackage[figure,table]{hypcap}
\else 
  \usepackage[dvips]{epsfig,graphicx}
  \usepackage[dvips,
    colorlinks=false, 
    bookmarksnumbered, 
    pdfstartview=XYZ 
  ]{hyperref}
  \usepackage[figure,table]{hypcap}
\fi

\title{\papertitle}

\affiliation{\paperauthorA}
{\href{http://users.informatik.uni-halle.de/~thielema/ResearchE.html}{Institut f\"ur Informatik} \\ Martin-Luther-Universit\"at Halle-Wittenberg \\ Halle, Germany\\
{\tt \href{mailto:henning.thielemann@informatik.uni-halle.de}{henning.thielemann@informatik.uni-halle.de}}
}

\usepackage{comment}
\usepackage{mathpaper}
\usepackage{IEEEtrantools}
\usepackage{nicefrac}

\newcommand\code[1]{\texttt{#1}}
\newcommand\term[1]{\textit{#1}}
\newcommand\haskell{\texttt{Haskell}}

\newcommand\periodicspace{\nicefrac{\R}{\Z}}
\newcommand\V{V} 
\newcommand\screwl[2]{#2 \boldsymbol\leftarrow #1}
\newcommand\screwr[2]{#2 \boldsymbol\rightarrow #1}
\newcommand\discretise[1]{Q{#1}}
\DeclareMathOperator{\sincone}{sinc1}
\DeclareMathOperator{\cisone}{cis1}
\DeclareMathOperator{\lerp}{lerp}
\DeclareMathOperator{\fractional}{frac}
\DeclareMathOperator{\round}{round}
\DeclareMathOperator{\toperiodic}{c}
\newcommand\fromperiodic{\toperiodic^{-1}}

\newcommand\figcaption[2]{\caption{\textit{#2}}\figlabel{#1}}
\newcommand\vcentergraphics[1]{
\begin{tabular}[c]{r}
\hspace{-5ex}
\includegraphics[width=0.9\columnwidth]{#1}
\end{tabular}
}

\graphicspath{{figures/}{program/figures/}}

\begin{document}
\ifpdf 
  \DeclareGraphicsExtensions{.png,.jpg,.pdf}
\else  
  \DeclareGraphicsExtensions{.eps}
\fi

\maketitle

\begin{abstract}
We are looking for a mathematical model of monophonic sounds
with independent time and phase dimensions.
With such a model we can resynthesise a sound
with arbitrarily modulated frequency and progress of the timbre.
We propose such a model and show
that it exactly fulfils some natural properties,
like a kind of time-invariance,
robustness against non-harmonic frequencies,
envelope preservation,
and inclusion of plain resampling as a special case.
The resulting algorithm is efficient and
allows to process data in a streaming manner
with phase and shape modulation at sample rate,
what we demonstrate with an implementation
in the functional language \haskell{}.
It allows a wide range of applications,
namely pitch shifting and time scaling,
creative FM synthesis effects,
compression of monophonic sounds,
generating loops for sampled sounds,
synthesise sounds similar to wavetable synthesis,
or making ultrasound audible.
\end{abstract}

\section{Introduction}

An example of our problem is illustrated in \figref{goal}.
Given is a signal of a monophonic sound of a known constant pitch.
We want to alter its pitch and the progression of its waveshape 
independently, possibly time-dependent, possibly rapidly.
The sound must not contain noise portions such as speech does.
We also do not try to preserve formants,
that is, like in resampling,
we accept that the spectrum of harmonics is stretched
by the same factor as the base frequency.
E.g. a square waveform shall remain square and so on.
For some natural instruments this is appropriate (e.g. guitar, piano)
whereas for other natural sounds this is inappropriate (e.g. speech).

The organisation of this article is inspired by
\cite{peyton-jones2004paper}.
With the paper we like to contribute the following:
\begin{itemize}
\item In \secref{problem} we specify our problem.
In \secref{cylinder-model}
we propose a mathematical model for monophonic sounds
given as real functions.
This model untangles phase and time
and allows us to describe frequency modulation and waveshape control.
In \secref{principle} we show how we utilise this model
for phase and time modification
and we formulate natural properties of this process.
\item \secref{cont-signal-theory}
is dedicated to theoretical details.
To this end we introduce some notations and definitions
in \secref{notation} and \secref{basic-functions}.
We investigate the properties from \secref{properties} like
time-invariance (\secref{time-invariance}),
linearity (\secref{linearity}),
preservation of static waves of the unit frequency (\secref{static-wave-preservation}),
preservation of pure sine waves and
robustness against non-harmonic frequencies (\secref{frequency-map}),
envelope preservation (\secref{envelope-preservation}),
inclusion of simple resampling and time warping
as a special case (\secref{resampling}),
and we prove that our model satisfies these properties exactly.
That is, our method is altogether theo\-re\-ti\-cally sound.
(I could not resist that pun!)
As bonus we verified some of the statements
using the proof assistant PVS in \secref{pvs-proofs}.
\item The problems of handling discrete signals
are treated in \secref{discrete-signal},
including notes on the implementation
in the purely functional programming language \haskell{}.
\item We suggest a range of applications of our method in \secref{application}.
\item In \secref{related-work} you find a survey of related work
and in \secref{example} we compare some results of our method
with the ones produced by the similar wavetable synthesis.
\item We finish our paper in \secref{outlook}
with a list of issues that we still need to work on.
\end{itemize}

\begin{figure}
\begin{tabular}{cc}
$x(t)$ & \vcentergraphics{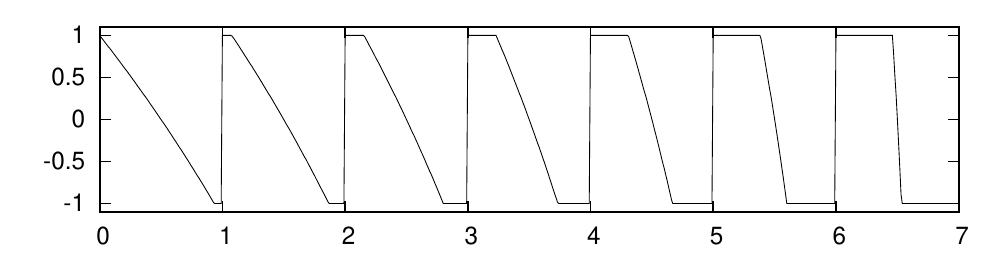} \\
$z(t)$ & \vcentergraphics{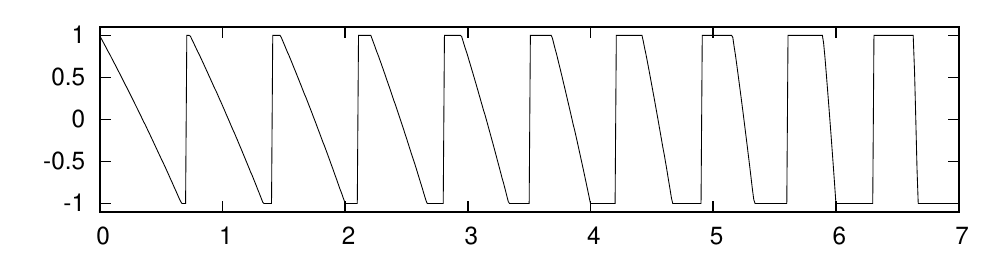} \\
 & time $t$
\end{tabular}
\figcaption{goal}{
A typical use case of our method:
From the above signal of a single tone
we want to compute the signal below.
That is, we want to alter the pitch
while maintaining the progression of its waveshape
and without knowing, how the signal was generated.
}
\end{figure}

\section{Continuous Signals: Overview}
\seclabel{cont-signal}

\subsection{Problem}
\seclabel{problem}

If we want to transpose a monophonic sound,
we could just play it faster for higher pitch
or slower for lower pitch.
This is how resampling works.
But this way the sound becomes also shorter or longer.
For some instruments like guitars this is natural,
but for other sounds like that of a brass, it is not necessarily so.
The problem we face is, that with ongoing time
both the waveform and the phase within the waveform change.
Thus we can hardly say, what the waveshape at a precise time point is.

If we could untangle phase and shape
this would open a wide range of applications.
We could independently control
progress of phase (i.e. frequency)
and progress of the waveshape.

\subsection{Model}
\seclabel{cylinder-model}

The wish for untangled phase and shape leads us straight forward
to the model we want to propose here.
If phase and shape shall be independent variables of a signal,
then our signal is actually a two-dimensional function,
mapping from phase and shape to the (particle) displacement.
Since the phase $\varphi$ is a cyclic quantity,
the domain of the signal function is actually a cylinder.
For simplicity we will identify the time point $t$ in a signal
with the shape parameter.
That is, in our model the time points to the instantaneous shape.

However, we never get signals in terms of a function on a cylinder.
So, how is this model related to real-word one-dimensional audio signals?
According to \figref{cylinder}
the easy direction is to get from the cylinder to the plain audio signal:
We move along the cylinder while increasing both the phase and shape parameter
proportionally to the time in the audio signal.
This yields a helical path.
The phase to time ratio is the frequency,
the shape to time ratio is the speed of shape progression.
The higher the ratio of frequency to shape progression,
the more dense the helix.
For constant ratio the frequency is proportional to the speed
with which we go along the helix.
We can change phase and shape non-proportionally to the time,
yielding non-helical paths.

When going from the one-dimensional signal to the two-dimensional signal,
there is a lot of freedom of interpretation.
We will use this freedom to make the theory as simple as possible.
E.g. we will assume, that the one-dimensional input signal
is an observation of the cylindrical function at a helical path.
Since we have no data for the function values beside the helix,
we have to guess them, in other words, we will interpolate.
\begin{figure}
\includegraphics[width=\columnwidth]{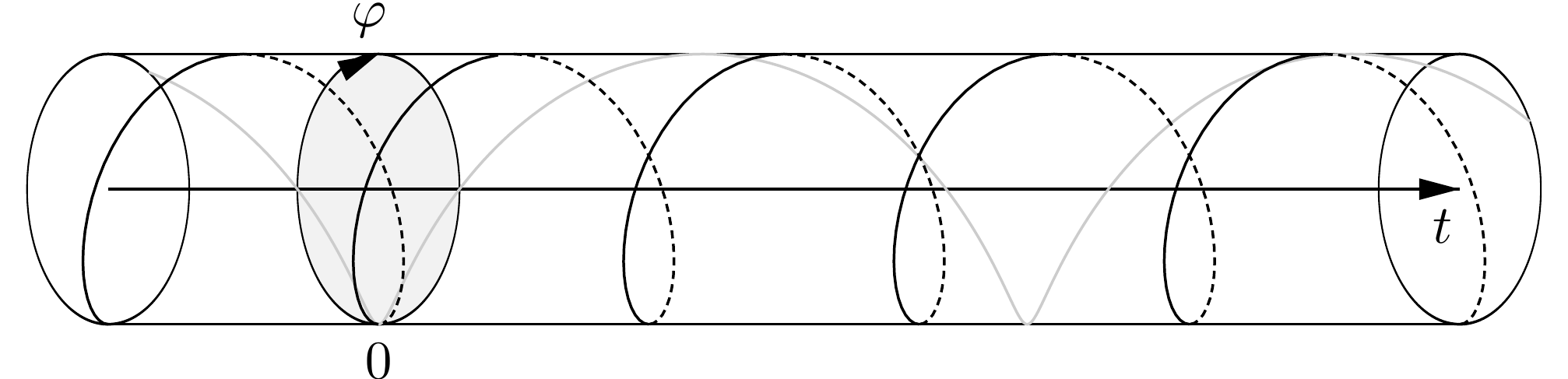}
\figcaption{cylinder}{
The cylinder we map the input signal onto (black and dashed helix)
and where we sample the output signal from (grey).
}
\end{figure}

This is actually a nice model
that allows us to perform many operations in an intuitive way
and thus it might be of interest beyond
pitch shifting and time scaling.

\subsection{Interpolation principle}
\seclabel{principle}

\newcommand{\tocyl}{F}
\newcommand{\fromcyl}{S}
\newcommand{\viacyl}{M}

An application of our model will
firstly cover the cylinder with data
that is interpolated from a one-dimensional signal $x$ by an operator $F$
and secondly it will choose some data
along a curve around that cylinder by an operator $S$.
The operator that we will work with here
has the structure
\[
\tocyl x(t,\varphi) = \sum_{k\in\Z} x(\varphi+k)\cdot\kappa(t-\varphi-k)
\]
where $\kappa$ is an interpolation kernel
such as a hat function or a sinus cardinalis ($\sinc$).
Intuitively spoken, it lays the signal on a helix on the cylinder.
Then on each line parallel to the time axis
there are equidistant discrete data points.
Now, $F$ interpolates them along the time direction
using the interpolation kernel $\kappa$.
You may check that $Fx(t,\varphi)$ has period 1 with respect to $\varphi$.
This is our way to represent the radian coordinate of the cylinder
within this section.

The observation operator $S$ shall sample along a helix
with time progression $v$ and angular speed $\alpha$:
\[
\fromcyl y (t) = y(v\cdot t, \alpha\cdot t)
\quad.
\]

Interpolation and observation together, yield
\begin{IEEEeqnarray*}{r/C/l}
\viacyl x (t)
 &=& \fromcyl (\tocyl x) (t) \\
 &=& \sum_{k\in\Z} x(\alpha\cdot t+k)\cdot\kappa((v-\alpha)\cdot t-k)
\quad.
\end{IEEEeqnarray*}

\seclabel{properties}

This operator turns out to have some useful properties:
\begin{enumerate}
\item Time-invariance \\
\label{item:time-invariance}
In audio signals often the absolute time is not important,
but the time differences.
Where you start an audio recording
should not have substantial effects on an operation you apply to it.
This is equivalent to the statement,
that a delay of the signal shall be mapped to a delayed result signal.
In particular it would be nice to have the property,
that a delay of the input by $v\cdot t$ yields a delay by $t$ of the output.
However this will not work.
To this end consider pure time-stretching ($\alpha=1$) applied to grains,
and we become aware that this property implies plain resampling,
which clearly changes the pitch.
What we have at least, is a restricted time invariance:
You have a discrete set of pairs of delays of input and output signal
that are mapped to each other
wherever the helices in \figref{cylinder} cross,
that is wherever $(v-\alpha)\cdot t \in\Z$.

However the construction $F$ of our \emph{model} is time invariant
in the sense
\begin{IEEEeqnarray*}{r/C/l}
x_1(t) &=& x_0(t-\tau) \\
\implies
\tocyl x_1(t,\varphi)
 &=& \tocyl x_0(t-\tau,\varphi-\tau)
\quad.
\eqnlabel{time-invariance-pointwise}
\end{IEEEeqnarray*}

\item Linearity \\
Since both $\tocyl$ and $\fromcyl$ are linear,
our phase and time modification process is linear as well.
This means that physical units and overall magnitudes
of signal values are irrelevant (\term{homogeneity})
and mixing before interpolation
is equivalent to mixing after interpolation (\term{additivity}).
\begin{IEEEeqnarray*}{s"r/C/l}
Homogeneity & \viacyl(\lambda\cdot x) &=& \lambda\cdot \viacyl x
   \eqnlabel{homogeneity} \\
Additivity & \viacyl(x+z) &=& \viacyl x + \viacyl z
   \eqnlabel{additivity}
\end{IEEEeqnarray*}

\item Resampling as special case \\
\label{item:resampling}
We think, that pitch shifting and time scaling by factor~1
should leave the input signal unchanged.
We also think, that resampling is the most natural answer
to pitch shifting and time scaling by the same factor $\alpha=v$.
For interpolating kernels,
that is $\kappa(0)=1, \forall j\in\Z\setminus\{0\}: \kappa(j) = 0$,
this actually holds.
\[
\viacyl x(t) = x(v\cdot t)
\]

\item Mapping of sine waves \\
\label{item:mapping-sines}
Our phase and time manipulation method
maps sine waves to sine waves
if the kernel is the sinus cardinalis normalised to integral zeros.
\[
\kappa(t) =
  \begin{cases}
      1&: t=0 \\
      \frac{\sin (t\cdot\pi)}{t\cdot\pi} &: \text{otherwise}
  \end{cases}
\]
Choosing this kernel means \person{Whittaker} interpolation.
Now we consider a complex wave of frequency $a$
as input for the phase and time modification.
\begin{IEEEeqnarray*}{r/C/l}
x(t) &=& \exp(2\pi\imag\cdot a \cdot t) \\
a &=& b+n \eqnlabel{decompose-frequency} \\
n &\in& \Z \\
b &\in& (-\tfrac{1}{2},\tfrac{1}{2}) \\
\viacyl x(t) &=& \exp(2\pi\imag\cdot (b\cdot v + n\cdot\alpha) \cdot t)
   \eqnlabel{frequency-mapping}
\end{IEEEeqnarray*}
Note that for $\fractional{a} = \frac{1}{2}$,
the \person{Whittaker} interpolation will diverge.
If $b=0$, that is the input frequency $a$ is integral,
then the time progression has no influence on the frequency mapping,
i.e.\ the input frequency $a$ is mapped to $\alpha\cdot a$.
We should try to fit the input signal as good as possible
to base frequency $1$ by stretching or shrinking,
since then all harmonics have integral frequency.

The fact, that sine waves are mapped to sine waves, implies,
that the effect of $\viacyl$ to a more complex tone
can be described entirely in frequency domain.
An example of a pure pitch shift is depicted in \figref{spectrum-shifted}.
The peaks correspond to the harmonics of the sound.
We see that the peaks are only shifted.
That is, the shape and width of each peak is maintained,
meaning that the envelope of each harmonic is the same after pitch shifting.
\begin{figure}
\begin{tabular}{l}
\includegraphics[width=\columnwidth]{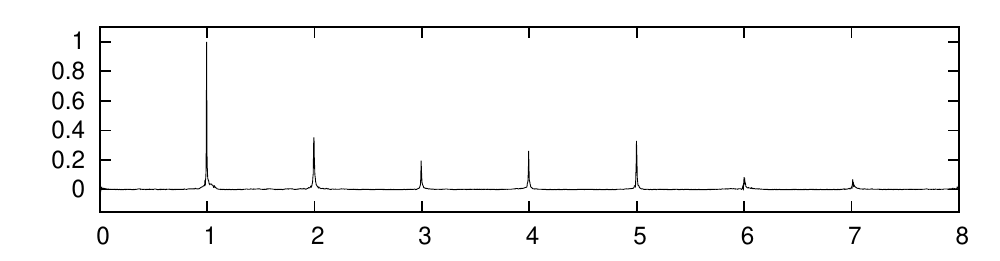} \\
\includegraphics[width=\columnwidth]{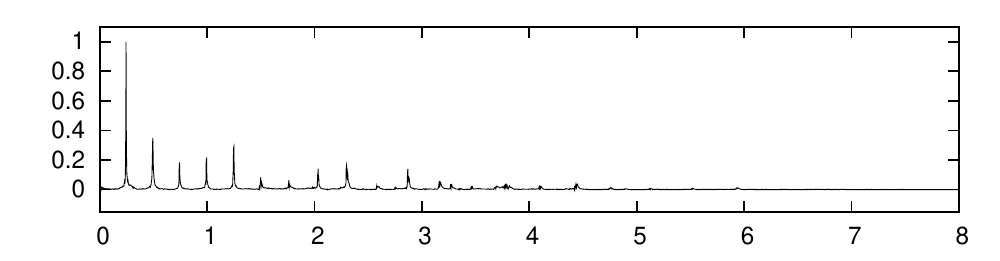}
\end{tabular}
\figcaption{spectrum-shifted}{
The first graph presents the lower part
of the absolute spectrum of a piano sound.
Its pitch is shifted 2 octaves down (factor~4) in the second graph.
}
\end{figure}

\item Preservation of envelope \\
Consider a static wave $x$,
i.e. $\forall t\ x(t)=x(t+1)$,
that is amplified according to an envelope $f$.
If interpolation with $\kappa$
is able to reconstruct $f$ and all of its translates
from their respective integral values,
then on the cylinder wave and envelope become separated
\[
\tocyl x(t,\varphi) = f(t)\cdot x(\varphi)
\]
and the overall phase and time manipulation algorithm
modifies frequency and time separately:
\[
\viacyl x(t) = f(v\cdot t)\cdot x(\alpha\cdot t)
\quad.
\]

Examples for $\kappa$ and $f$ are:
\begin{itemize}
\item
$\kappa$ being the sinus cardinalis
as defined in item \ref{item:mapping-sines} and
$f$ being a signal bandlimited to $(-\frac{1}{2},\frac{1}{2})$,
\item $\kappa = \charfunc{(-1,0]}$ and $f$ being constant,
\item $\kappa(t) = \max(0,1-\abs{t})$ and $f$ being a linear function,
\item $\kappa$ being an interpolation kernel,
that preserves polynomial functions up to degree~$n$ and
$f$ being such a polynomial function.
\end{itemize}

\end{enumerate}

\section{Continuous Signals: Theory}
\seclabel{cont-signal-theory}


In this section we want to give proofs of the statements
found in \secref{cont-signal}
and we want to check what we could have done alternatively
given the properties that we found to be useful.
You can safely skip the entire section
if you are only interested in practical results and applications.

\subsection{Notation}
\seclabel{notation}

In order to give precise, concise, even intuitive proofs,
we want to introduce some notations.

In signal processing literature
we find often a term like $x(t)$ being called a signal,
although from the context you derive,
that actually $x$ is the signal and
thus $x(t)$ denotes a displacement value of that signal at time $t$.
We like to be more strict in our paper.
We like to talk about signals as objects
without always going down to the level of single signal values.
Our notation should reflect this
and should clearly differentiate between signals and signal values.
This way, we can e.g. express a statement like
``delay and convolution commute'' by
\[
\translater{t}{(x*y)} = x*(\translater{t}{y})
\]
(cf. \eqnref{convolution-translation-associative})
which would be more difficult in a pointwise and correct~(!) notation.

This notation is inspired by functional programming,
where functions that process functions are called higher-order functions.
It allows us to translate the theory described here
almost literally to functional programs and theorem prover modules.
Actually some of the theorems stated in this paper
have been verified using PVS \cite{owre2001pvs}.
For a more detailed discussion of the notation,
see \cite{thielemann2006matchedwavelets}.

In our notation function application
has always higher precedence than infix operators.
Thus $\translater{t}{\discretise{x}}$ means
$\translater{t}{(\discretise{x})}$
and not $\discretise{(\translater{t}{x})}$.
Function application is left associative,
that is, $\discretise{x}(t)$ means $(\discretise{x})(t)$
and not $\discretise{(x(t))}$.
This is also the convention in Functional Analysis.
We use anonymous functions, also known as lambda expressions.
The expression $\anonymfunc{x}{Y}$
denotes a function $f$ where $\forall x\ f(x) = Y$ and
$Y$ is an expression that usually contains $x$.
Arithmetic infix operators like ``$+$'' and ``$\cdot$''
shall have higher precedence than the mapping arrow,
and logical infix operators like ``$=$'' and ``$\land$''
shall have lower precedence.
That is, $\anonymfunc{t}{f(t-\tau)} = \translater{\tau}{f}$
means $(\anonymfunc{t}{(f(t-\tau)+g(t-\tau))}) = (\translater{\tau}{(f+g)})$.

\begin{definition}[Function set]
With \[\funcset{A}{B}\] we like to denote the set of all functions
mapping from set~$A$ to set~$B$.
This operation is treated right associative,
that is, $\funcset{A}{\funcset{B}{C}}$
means $\funcset{A}{(\funcset{B}{C})}$,
not $\funcset{(\funcset{A}{B})}{C}$.
This convention matches the convention of left associative function application.
\end{definition}

\subsection{Basic functions}
\seclabel{basic-functions}

For the description of the cylinder
we first need the notion of a cyclic quantity.
\begin{definition}[Cyclic quantity]
Intuitively spoken,
cyclic (or periodic) quantities are values in the range $[0,1)$
that wrap around at the boundaries.
More precisely, a cyclic quantity $\varphi$ is a set of real numbers
that all have the same fractional part.
Put differently, a periodic quantity is an equivalence class
with respect to the relation,
that two numbers are considered equivalent
when their difference is integral.
In terms of a quotient space this can concisely be written as
\[ \varphi\in\periodicspace \quad.\]
\end{definition}

\begin{definition}[Periodisation]
Periodisation $\toperiodic$ means mapping a real value to a cyclic quantity,
i.e. choosing the equivalence class belonging to a representative.
\begin{eqnarray*}
\toperiodic &\in& \funcset{\R}{\periodicspace} \\
\forall p\in\R \quad \toperiodic(p)
   &=& p+\Z \\
   &=& \{q : q-p \in\Z \}
\end{eqnarray*}
It holds $\toperiodic(0)=\Z$.
We define the inverse of $\toperiodic$
as picking a representative from the range $[0,1)$.%
\begin{eqnarray*}
\fromperiodic &\in& \funcset{\periodicspace}{\R} \\
\forall \varphi\in\periodicspace \quad
   \fromperiodic(\varphi) &\in& \setsection{\varphi}{[0,1)}
\end{eqnarray*}
\end{definition}
In a computer program,
we do not encode the elements of $\periodicspace$ by sets of numbers,
but instead we store a representative between 0 and 1,
including 0 and excluding 1.
Then $\toperiodic$ is just the function, that computes the fractional part,
i.e.~\texttt{c~t~= t - floor t}.

A function $y$ on the cylinder is thus from
$\funcset{(\R\times\periodicspace)}{\V}$,
where $\V$ denotes a vector space.
E.g. for $\V=\R$ we have a mono signal,
for $\V=\R\times\R$ we obtain a stereo signal and so on.

The conversion $S$ from the cylinder to an audio signal
is entirely determined by given
phase control curve $g$ and shape control curve $h$.
It consists of picking the values from the cylinder along the path
that corresponds to these control curves.
\begin{eqnarray}
S_{h,g} &\in&
 \funcset{(\funcset{(\R\times\periodicspace)}{\V})}{(\funcset{\R}{\V})} \\
S_{h,g}y(t) &=& y(h(t),g(t)) \eqnlabel{cylinder-to-audio}
\end{eqnarray}

For the conversion $F$ from a prototype audio signal to a cylindrical model
we have a lot of freedom.
In section \secref{principle} we have seen what properties
a certain $F$ has, that we use in our implementation.
We will going on to check what choices for $F$ we have,
given that these properties hold.
For now we will just record, that
\begin{eqnarray*}
F &\in& \funcset{(\funcset{\R}{\V})}{(\funcset{(\R\times\periodicspace)}{\V})}
\quad.
\end{eqnarray*}

\subsection{Properties}


\subsubsection{Time-Invariance}
\seclabel{time-invariance}

\begin{definition}[Translation, Rotation]
Shifting a signal $x$ forward or backward in time
or rotating a waveform with respect to its phase
shall be expressed by an intuitive arrow notation
that is inspired by \cite{strang1996cascade,sweldens1998liftingfactor}
and was already successfully applied in \cite{thielemann2006matchedwavelets}:
\begin{eqnarray}
(\translater{\tau}{x})(t) &=& x(t-\tau) \eqnlabel{translate-right} \\
(\translatel{\tau}{x})(t) &=& x(t+\tau) \eqnlabel{translate-left}
\quad.
\end{eqnarray}
For a cylindrical function we have two directions,
one for rotation and one for translation.
We define analogously
\begin{eqnarray}
(\screwr{(\tau,\alpha)}{y})(t, \varphi) &=& y(t-\tau, \varphi-\alpha) \eqnlabel{screw-right} \\
(\screwl{(\tau,\alpha)}{y})(t, \varphi) &=& y(t+\tau, \varphi+\alpha) \eqnlabel{screw-left}
\quad.
\end{eqnarray}
\end{definition}

The first notion of time-invariance that comes to mind,
can be easily expressed using the arrow notation by
$\forall t\ F(\translater{t}{x}) = \translater{(t,c(0))}{Fx}$.
However, this will not yield any useful conversion.
Shifting the time always includes shifting the phase
and our notion of time-invariance must respect that.
We have already given an according definition in
\eqnref{time-invariance-pointwise}
that we can now write using the arrow notation.
\begin{definition}[Time-invariant cylinder interpolation]
\dfnlabel{time-invariance}
We call an interpolation operator $F$
time-invariant whenever it satisfies
\begin{eqnarray}
\forall x\ \forall t & &
   F(\translater{t}{x}) = \screwr{(t,\toperiodic(t))}{Fx}
\quad.
\eqnlabel{time-invariance}
\end{eqnarray}
\end{definition}
Using this definition,
we do not only force $F$ to map translations to translations,
but we also fix the factor of the translation distance to $1$.
That is, when shifting an input signal $x$,
the according model $Fx$ is shifted along the unit helix,
that turns once per time difference 1.

Enforcing the time-invariance property
restricts our choice of $F$ considerably.
\begin{IEEEeqnarray*}{r/C/l"l}
\separateeqn{Fx(t, \varphi)}
  &=& (\screwl{(t,\toperiodic(t))}{Fx})(0, \varphi-\toperiodic(t))
         & \eqnremark{\eqnref{screw-left}} \\
  &=& F(\translatel{t}{x})(0, \varphi-\toperiodic(t))
         & \eqnremark{\eqnref{time-invariance}}
\end{IEEEeqnarray*}
We see, that actually only a ring slice of $F(\translatel{t}{x})$
at time point zero is required
and we can substitute
$Ix(\varphi) = Fx(0,\varphi)$.
$I$ is an operator
from $\funcset{(\funcset{\R}{\V})}{(\funcset{\periodicspace}{\V})}$,
that turns a straight signal into a waveform.
Now we know, that time-invariant interpolations can only be of the form
\begin{IEEEeqnarray*}{r/C/l}
Fx(t,\varphi) &=& I(\translatel{t}{x})(\varphi-\toperiodic(t))
\eqnlabel{time-invariant-ip-pointwise} \\
\eqntext{or more concisely}
\anonymfunc{\varphi}{Fx(t,\varphi)}
 &=& \translater{\toperiodic(t)}{I(\translatel{t}{x})}
\eqnlabel{time-invariant-ip}
\quad.
\end{IEEEeqnarray*}
The last line can be read as:
In order to obtain a ring slice of the cylindrical model at time $t$,
we have to move the signal, such that time point $t$ becomes point $0$,
then apply $I$ to get a waveform on a ring,
then rotate back that ring correspondingly.

We may check, that any $F$ defined this way is indeed time-invariant
in the sense of \eqnref{time-invariance}.
\begin{IEEEeqnarray*}{r/C/l"l}
\separateeqn{F(\translater{t}{x})(\tau,\varphi)}
 &=& I(\translatel{\tau}{(\translater{t}{x})})(\varphi-\toperiodic(\tau))
         & \eqnremark{\eqnref{time-invariant-ip-pointwise}} \\
 &=& I(\translatel{(\tau-t)}{x})(\varphi-\toperiodic(\tau)) \\
 &=& I(\translatel{(\tau-t)}{x})(\varphi-\toperiodic(t)-\toperiodic(\tau-t)) \\
 &=& Fx(\tau-t,\varphi-\toperiodic(t))
         & \eqnremark{\eqnref{time-invariant-ip-pointwise}} \\
 &=& (\screwr{(t,\toperiodic(t))}{Fx})(\tau,\varphi)
\end{IEEEeqnarray*}

\subsubsection{Linearity}
\seclabel{linearity}

We like that our phase and time modification process is linear
(as in \eqnref{homogeneity} and \eqnref{additivity}).
Since sampling $S$ from the cylinder is linear,
the interpolation $F$ to the cylinder must be linear as well.
\begin{IEEEeqnarray*}{s"r/C/l}
Homogeneity & F(\lambda\cdot x) &=& \lambda\cdot Fx \\
Additivity & F(x+z) &=& Fx + Fz
\end{IEEEeqnarray*}
The properties of $F$ are equivalent to
\begin{eqnarray*}
I(\lambda\cdot x) &=& \lambda\cdot Ix \\
I(x+z) &=& Ix + Iz
\quad.
\end{eqnarray*}

\subsubsection{Static wave preservation}
\seclabel{static-wave-preservation}

Another natural property is,
that an input signal consisting of a wave of constant shape
is mapped to the cylinder where each ring contains that waveform.
A static waveform can be written concisely as $\compose{w}{\toperiodic}$.
It denotes the function composition of $w$ and $\toperiodic$,
that is, $w$ is applied to the result of $\toperiodic$,
for example $(\compose{w}{\toperiodic})(2.3) = w(c(0.3))$.
Thus $w$ and $\compose{w}{\toperiodic}$ both represent periodic functions,
but $w$ has domain $\periodicspace$ and thus is periodic by its type,
whereas $\compose{w}{\toperiodic}$ is an ordinary real function,
that happens to satisfy the periodicity property
$(\compose{w}{\toperiodic}) = \translater{1}{(\compose{w}{\toperiodic})}$.
%
%
We can write our requirement as
\[
\forall t\ \forall \varphi\quad
F(\compose{w}{\toperiodic})(t,\varphi) = w(\varphi)
\quad.
\]

As an example we have a constant interpolation
\begin{IEEEeqnarray*}{r/C/l}
I(x) &=& \compose{x}{\fromperiodic} \\
Fx(t,\varphi) &=& x\left(t+\fromperiodic(\varphi-\toperiodic(t))\right)
\quad.
\end{IEEEeqnarray*}
We illustrate the constant interpolation in \figref{constant-interpolation},
but with a sine wave, that does not have frequency 1,
and thus looks for the interpolation operator $F$ like a non-static waveform.
This way, we can better demonstrate how constant interpolation works,
and we think one can verify intuitively, how it preserves static waves.
\begin{figure}
\includegraphics[width=\columnwidth]{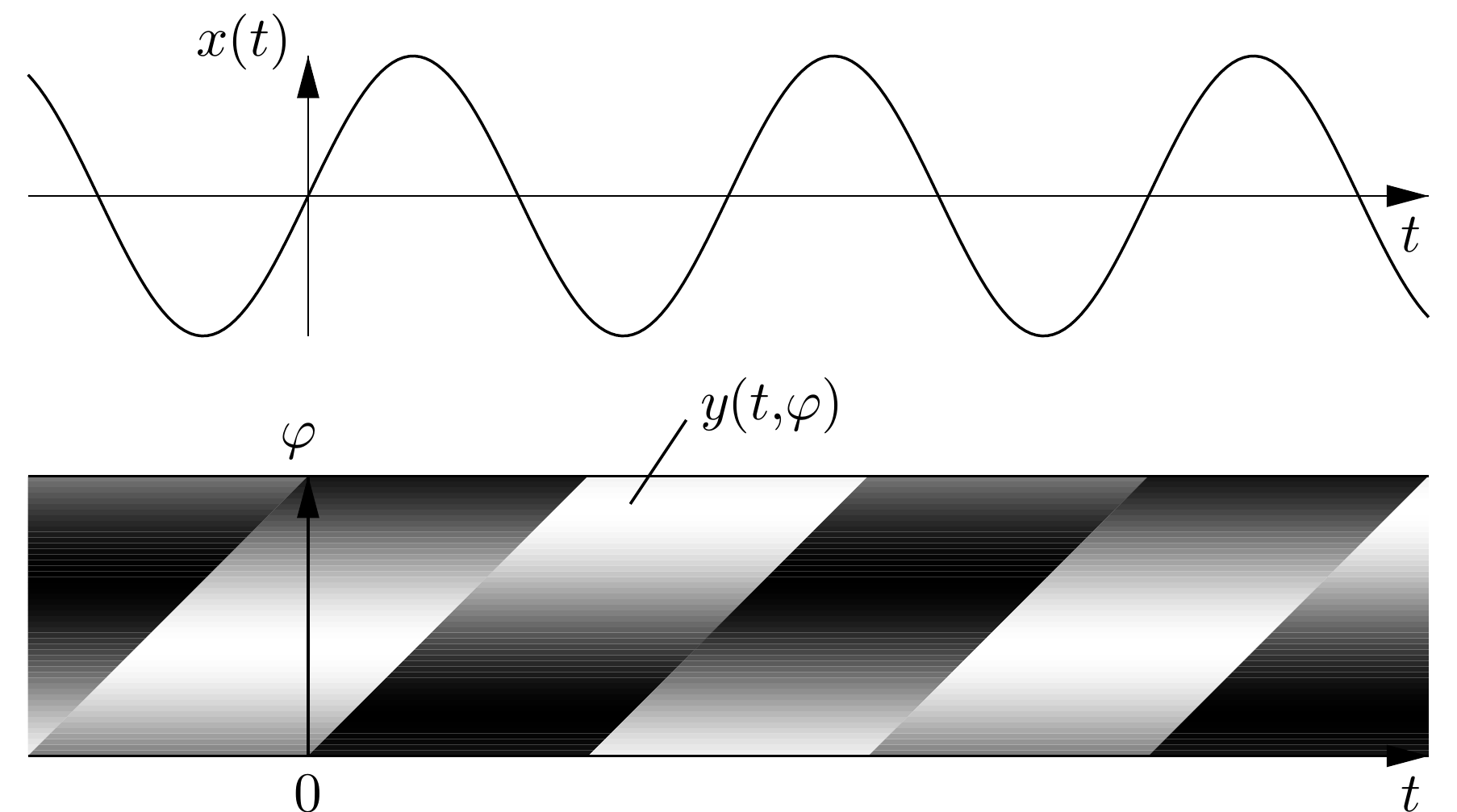}
\figcaption{constant-interpolation}{
Constant interpolation (below) of a sine wave (above) that is out of sync.
The interpolation picture represents the surface of the cylinder
after cutting and flattening.
A black dot means $y(t,\varphi) = -1$ and a white dot represents $1$.
The sine wave can be found in the interpolation image
at the right border of each of the skew stripes.
Along the vertical line from bottom to top
you find the first period of the input signal,
where ``first'' is measured from time point $0$.
}
\end{figure}

We can consider an input signal of the form $\compose{w}{\toperiodic}$
as a wave with constant envelope
and we will generalise this to other envelopes in
\secref{envelope-preservation}.

\subsubsection{Mapping of pure sine waves}
\seclabel{frequency-map}

We like to derive, how frequencies are mapped
when converting from an audio signal to the cylindrical model
and observing the signal along a different but uniform helix.
To this end, we need an interpolation that maps sine waves to sine waves.
Actually, the \person{Whittaker} interpolation has this property.
\begin{IEEEeqnarray*}{r/C/l}
\sincone t
 &=& \lim_{\tau\to t} \frac{\sin (\tau\cdot\pi)}{\tau\cdot\pi} \\
 &=& \begin{cases}
         1&: t=0 \\
         \frac{\sin (t\cdot\pi)}{t\cdot\pi} &:\text{otherwise}
     \end{cases}
\\
Fx(t,\varphi)
   &=& \sum_{\tau\in\varphi} x(\tau)\cdot\sincone(t-\tau)
\eqnlabel{whittaker-interpolation}
\end{IEEEeqnarray*}
Since $\varphi\in\periodicspace$,
when $\tau\in\varphi$
then $\tau$ assumes all values that differ from $\fromperiodic(\varphi)$ by an integer.
The infinite sum $\sum_{\tau\in\varphi} f(\tau)$
shall be understood as
$\lim_{n\to\infty}\sum_{\tau\in\setsection{\varphi}{[-n,n]}} f(\tau)$.

The proof of $F$ being time-invariant according to \dfnref{time-invariance}
is deferred to \secref{kernel-interpolation-time-invariant},
where we perform the proof for any interpolating kernel,
not just $\sincone$.

We will now demonstrate, that $\sincone$-interpolation preserves sine waves
and how frequencies are mapped.

\paragraph*{Mapping a complex sine wave to the cylinder}

Since exponential laws are much easier to cope with
than addition theorems for sine and cosine,
we use a complex wave defined by
\begin{eqnarray*}
\cisone t &=& \exp(2\pi\imag\cdot t)
\quad.
\end{eqnarray*}
For the following derivation we need the
\person{Whittaker}-\person{Shannon} interpolation formula
\cite{hamming1989digitalfilter} in the form
\begin{multline}
\forall b \in (-\tfrac{1}{2},\tfrac{1}{2}) \\
\sum_{k\in\Z} \cisone(b\cdot k)\cdot\sincone(t-k)
 = \cisone(b\cdot t)
.
\eqnlabel{whittaker-interpolation-of-sine}
\end{multline}
We choose a complex wave of frequency $a$
as input for the conversion to the cylinder.
The fractional frequency part $b$ and the integral frequency $n$
are chosen as in \eqnref{decompose-frequency}.
\begin{IEEEeqnarray*}{r/C/l}
x(t) &=& \cisone(a\cdot t) \\
\text{with\quad}
a &=& b+n \\
n &\in& \Z \\
b &\in& (-\tfrac{1}{2},\tfrac{1}{2})
\end{IEEEeqnarray*}
This choice implies the following interpolation result
\begin{IEEEeqnarray*}{r/C/l}
Fx(t,\varphi)
   &=& \sum_{\tau\in\varphi} \cisone(a\cdot\tau)\cdot\sincone(t-\tau) \\
\forall \tau\in\varphi \qquad \\
Fx(t,\varphi)
   &=& \cisone(a\cdot\tau)\cdot\sum_{k\in\Z} \cisone(a\cdot k)\cdot\sincone(t-\tau-k) \\
\eqntext{because $a-b\in\Z$}
   &=& \cisone(a\cdot\tau)\cdot\sum_{k\in\Z} \cisone(b\cdot k)\cdot\sincone(t-\tau-k) \\
   &=& \cisone(a\cdot\tau)\cdot \cisone(b\cdot(t-\tau))
      \eqncomment{\eqnref{whittaker-interpolation-of-sine}} \\
   &=& \cisone(b\cdot t+n\cdot\tau) \\
Fx(t,\varphi)
   &=& \cisone\left(b\cdot t + n\cdot\fromperiodic(\varphi)\right)
\eqnlabel{sine-preserved}
\quad.
\end{IEEEeqnarray*}
The result can be viewed in \figref{whittaker-interpolation}.
\begin{figure}
\includegraphics[width=\columnwidth]{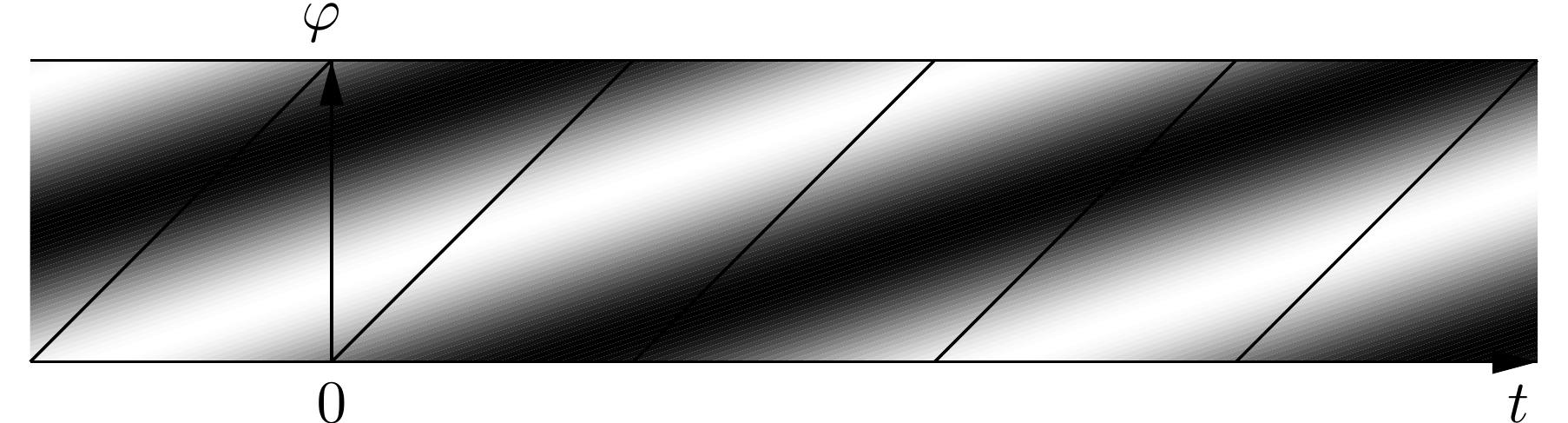}
\figcaption{whittaker-interpolation}{
The sine wave as in \figref{constant-interpolation}
is interpolated by \person{Whittaker} interpolation.
Along the diagonal lines you find the original sine wave.
}
\end{figure}
We obtain, that for every $t$
the function on a ring slice $\anonymfunc{\varphi}{Fx(t,\varphi)}$ is a sine wave
with the integral frequency $n$ that is closest to $a$.
That is, the closer $a$ is to an integer,
the more harmonics of a non-sine wave are mapped to corresponding harmonics
in a ring slice of $Fx$.

\paragraph*{Mapping a complex wave from the cylinder to an audio signal}

For time progression speed $v$ and frequency $\alpha$ we get
\begin{IEEEeqnarray*}{r/C/l}
z(t) &=& Fx(v\cdot t,\toperiodic(\alpha\cdot t)) \\
 &=& \cisone(b\cdot v\cdot t + n\cdot\fromperiodic(\toperiodic(\alpha\cdot t))) \\
\eqntext{because $\forall \tau\in\R\quad\fromperiodic(\toperiodic(\tau))-\tau \in \Z$}
 &=& \cisone(b\cdot v\cdot t + n\cdot\alpha\cdot t) \\
 &=& \cisone((b\cdot v + n\cdot\alpha)\cdot t)
\quad.
\end{IEEEeqnarray*}
This proves \eqnref{frequency-mapping}.


\subsubsection{Interpolation using kernels}
\seclabel{kernel-interpolation}

Actually, for the two-dimensional interpolation $F$
we can use any interpolation kernel $\kappa$,
not only $\sincone$ as in \eqnref{whittaker-interpolation}.
\begin{eqnarray}
%
%
Fx(t,\varphi)
   &=& \sum_{\tau\in\varphi} x(\tau)\cdot\kappa(t-\tau)
\eqnlabel{kernel-interpolation-pointwise}
\end{eqnarray}
The constant interpolation corresponds to $\kappa = \charfunc{(-1,0]}$.
Linear interpolation is achieved using a hat function.

\seclabel{kernel-interpolation-time-invariant}
\begin{lemma}[Time invariance of kernel interpolation]
The operator $F$ defined with an interpolation kernel as in
\eqnref{kernel-interpolation-pointwise}
is time-invariant according to \dfnref{time-invariance}.
\end{lemma}

\begin{proof}
\begin{eqnarray*}
F(\translater{d}{x})(t,\varphi)
   &=& \sum_{\tau\in\varphi} (\translater{d}{x})(\tau)\cdot\kappa(t-\tau) \\
   &=& \sum_{\tau\in\varphi} x(\tau-d)\cdot\kappa((t-d)-(\tau-d)) \\
   &=& \sum_{\tau\in(\varphi-\toperiodic(d))} x(\tau)\cdot\kappa(t-d-\tau) \\
   &=& (\screwr{(d,\toperiodic(d))}{Fx})(t,\varphi)
\end{eqnarray*}
\end{proof}

Conversely, we like to note,
that kernel interpolation is not the most general form
when we only require time-invariance, linearity and static wave preservation.

The following considerations are simplified
by rewriting general kernel interpolation to a more functional style
using a discretisation operator and a mixed discrete/continuous convolution.

\begin{definition}[Quantisation]
With quantisation we mean the operation
that picks the signal values at integral time points from a continuous signal.
\begin{IEEEeqnarray*}{r/C/l}
\discretise{} &\in& \funcset{(\funcset{\R}{\V})}{(\funcset{\Z}{\V})} \\
\forall n\in\Z\quad
  \discretise{x}(n) &=& x(n)
\eqnlabel{definition-quantisation}
\end{IEEEeqnarray*}
\end{definition}
Here is, how quantisation operates on pointwise multiplied signals
and on periodic signals:
\begin{eqnarray}
\discretise{(x\cdot z)} &=& \discretise{x}\cdot \discretise{z}
 \eqnlabel{discretise-product} \\
\forall n\in\Z\quad
\discretise{(\compose{w}{\toperiodic})}(n) &=& w(\toperiodic(0))
 \eqnlabel{discretise-periodic}
\quad.
\end{eqnarray}

\begin{definition}[Mixed Convolution]
For $u\in\funcset{\Z}{\V}$ and $x\in\funcset{\R}{\R}$
then mixed discrete/continuous convolution is defined by
\begin{eqnarray*}
(u * x)(t) &=& \sum_{k\in\Z} u(k)\cdot x(t-k)
\end{eqnarray*}
\end{definition}
We can express mixed convolution
also by purely discrete convolutions:
\begin{eqnarray*}
\discretise(\translatel{t}{(u * x)})
   &=& u * \discretise(\translatel{t}{x})
\eqnlabel{mixed-convolution-as-discrete-convolution}
\quad.
\end{eqnarray*}

It holds
\begin{eqnarray}
\translater{t}{(u*x)} &=& u*(\translater{t}{x})
,
\eqnlabel{convolution-translation-associative}
\end{eqnarray}
because translation can be written as convolution
with a translated \person{Dirac} impulse
and convolution is associative in this case
(and generally when infinity does not cause problems).
Thus we will omit the parentheses.
We like to note,
that this example demonstrates the usefulness of the functional notation,
since without it even a simple statement like
\eqnref{convolution-translation-associative}
is hard to formulate in a correct and unambiguous way.

These notions allow us to rewrite kernel interpolation
\eqnref{kernel-interpolation-pointwise}:
\begin{IEEEeqnarray*}{r/C/l}
\forall \tau\in\varphi\quad
Fx(t,\varphi)
   &=& \sum_{k\in\Z} x(k+\tau)\cdot\kappa(t-(k+\tau)) \\
\forall \tau\in\varphi\quad
\anonymfunc{t}{Fx(t,\varphi)}
   &=& \discretise{\left(\translatel{\tau}{x}\right)}
         * \translater{\tau}{\kappa}
\quad.
\eqnlabel{kernel-interpolation}
\end{IEEEeqnarray*}
The last line can be read as follows:
The signal on the cylinder along a line parallel to the time axis
can be obtained by taking discrete points of $x$
and interpolate them using the kernel $\kappa$.

\subsubsection{Envelope preservation}
\seclabel{envelope-preservation}

We can now generalise the preservation of static waves from
\secref{static-wave-preservation}
to envelopes different from a constant function.
\begin{lemma}
\lemlabel{envelope-preservation}
Given an envelope $f$ from $\funcset{\R}{\R}$
and an interpolation kernel $\kappa$
that preserves any translated version of $f$, i.e.
\begin{eqnarray}
\forall t\quad \discretise{(\translatel{t}{f})} * \kappa &=& \translatel{t}{f},
\eqnlabel{preserve-translated-envelope}
\end{eqnarray}
then and only then,
a wave of constant shape $w$ enveloped by $f$
is converted to constant waveshapes on the cylinder rings
enveloped by $f$ in time direction:
\begin{eqnarray}
F(f\cdot (\compose{w}{\toperiodic}))(t,\varphi)
 &=& f(t)\cdot w(\varphi)
\quad.
\eqnlabel{convert-envelope-wave}
\end{eqnarray}
\end{lemma}

\begin{proof}
\begin{eqnarray*}
\noalign{$\forall \tau\in\varphi\quad
\anonymfunc{t}{F(f\cdot (\compose{w}{\toperiodic}))(t,\varphi)}$}
   &=& \discretise{\left(\translatel{\tau}
               {(f\cdot (\compose{w}{\toperiodic}))}\right)}
         * \translater{\tau}{\kappa} \\
   &=& \discretise{\left((\translatel{\tau}{f})\cdot
               (\compose{(\translatel{\varphi}{w})}{\toperiodic})\right)}
         * \translater{\tau}{\kappa} \\
   &=& w(\varphi)\cdot\discretise{(\translatel{\tau}{f})}
         * \translater{\tau}{\kappa}
      \eqncomment{(\ref{eqn:discretise-product},\ref{eqn:discretise-periodic},\ref{eqn:translate-left})}
\end{eqnarray*}
Now the implication
 $\eqnref{preserve-translated-envelope} \implies
     \eqnref{convert-envelope-wave}$
should be obvious,
whereas the converse 
 $\eqnref{convert-envelope-wave} \implies
     \eqnref{preserve-translated-envelope}$
can be verified by setting $\forall \varphi\ w(\varphi) = 1$.
This special case means
that the envelope $f$ used as input signal is preserved in the sense
\begin{eqnarray*}
Ff(t,\varphi) &=& f(t)
\quad.
\end{eqnarray*}

\end{proof}

\begin{corollary}
When we convert back to a one-dimensional audio signal
under the condition \eqnref{preserve-translated-envelope},
then the time control only affects the envelope
and the phase control only affects the pitch:
\[
S_{h,g} (F(f\cdot(\compose{w}{\toperiodic}))) =
(\compose{f}{h})\cdot(\compose{w}{g})
\quad.
\]
\end{corollary}

\subsubsection{Special cases}
\seclabel{resampling}

As stated in item~\ref{item:resampling} of \secref{principle}
we like to have resampling as special case
of our phase and time manipulation algorithm.
It turns out, that this property is equivalent
to putting the input signal $x$ on the diagonal lines
as in \figref{constant-interpolation} and \figref{whittaker-interpolation}.
We will derive, what this imposes on the choice of the kernel $\kappa$
when $F$ is defined via a kernel as in \eqnref{kernel-interpolation}.
\begin{lemma}
For $F$ defined by
\[
\forall \tau\in\varphi\quad
\anonymfunc{t}{Fx(t,\varphi)}
   = \discretise{\left(\translatel{\tau}{x}\right)}
         * \translater{\tau}{\kappa}
\]
it holds
\begin{eqnarray}
\forall x\ \forall t\in\R \quad
x(t) &=& Fx(t,\toperiodic(t))
\eqnlabel{signal-on-unit-helix}
\end{eqnarray}
if and only if
\[
\discretise{\kappa} = \delta
,
\]
that is,
$\kappa$ is a so called interpolating kernel.

Here, $\delta$ is the discrete \person{Dirac} impulse,
that is
\[
\forall k\in\Z\quad
\delta(k) = \begin{cases} 1 &: k=0 \\ 0 &: \text{otherwise} \end{cases}
.
\]
\end{lemma}

\begin{proof}

``$\Rightarrow$''
\begin{IEEEeqnarray*}{r/C/l}
\forall x\ \forall t\in\R \quad
x(t)
 &=& Fx(t,\toperiodic(t)) \\
 &=& (\discretise{(\translatel{t}{x})} * \translater{t}{\kappa})(t) \\
\eqntext{consider only $t\in\Z$ and rename it to $k$}
\forall x\ \forall k\in\Z \quad
x(k)
 &=& (\discretise{(\translatel{k}{x})} * \translater{k}{\kappa})(k) \\
 &=& (\discretise{x} * \kappa)(k) \\
\forall x \quad
\discretise{x}
 &=& \discretise{(\discretise{x} * \kappa)} \\
 &=& \discretise{x} * \discretise{\kappa}
         \qquad\text{(discrete convolution)}
\quad.
\end{IEEEeqnarray*}
For $\discretise{x} = \delta$
we get $\delta = \delta*\discretise{\kappa} = \discretise{\kappa}$.

\bigskip
\noindent ``$\Leftarrow$'' \\
Conversely, every interpolating kernel $\kappa$
asserts \eqnref{signal-on-unit-helix}:
\begin{IEEEeqnarray*}{r/C/l"l}
\noalign{$ \forall x\ \forall t\in\R \quad
(\discretise{(\translatel{t}{x})} * \translater{t}{\kappa})(t) $}
\qquad\qquad\qquad\qquad
 &=& (\discretise{(\translatel{t}{x})} * \kappa)(0)
        & \eqnremark{\eqnref{translate-right}} \\
 &=& \discretise{(\discretise{(\translatel{t}{x})} * \kappa)}(0)
        & \eqnremark{\eqnref{definition-quantisation}} \\
 &=& (\discretise{(\translatel{t}{x})} * \discretise{\kappa})(0)
        & \eqnremark{\eqnref{mixed-convolution-as-discrete-convolution}} \\
 &=& (\discretise{(\translatel{t}{x})} * \delta)(0) \\
 &=& (\translatel{t}{x})(0) \\
 &=& x(t)
\quad.
\end{IEEEeqnarray*}

\end{proof}

Now, when our conversion from the cylinder to the one-dimensional signal
does only walk along the unit helix,
we get general time warping as special case of our method:
\begin{IEEEeqnarray*}{r/C/l"l}
S_{h,\compose{\toperiodic}{h}}(Fx)
 &=& \anonymfunc{t}{Fx(h(t), \toperiodic(h(t)))}
       & \eqnremark{\eqnref{cylinder-to-audio}} \\
 &=& \anonymfunc{t}{x(h(t))}
       & \eqnremark{\eqnref{signal-on-unit-helix}} \\
 &=& \compose{x}{h}
\end{IEEEeqnarray*}
For $h=\id$ we get the identity mapping,
for $h(t) = v\cdot t$ we get resampling by speed factor~$v$.

\section{Discrete Signals}
\seclabel{discrete-signal}

For the application of our method to sampled signals
we could interpolate a discrete signal $u$
containing a wave with period $T$,
thus getting a continuous signal $x$ with $x(\frac{n}{T}) = u(n)$
and proceed with the technique for continuous signals from \secref{cont-signal}.
However, when working out the interpolation
this yields a skew grid with two alternating cell heights
and a doubled number of parallelogram cells,
which seems to be unnatural to us.
Additionally it would require three distinct interpolations,
e.g. two distinct interpolations in the unit helix direction
and one interpolation in time direction.
Instead we want to propose a periodic scheme
where we need two interpolations with the same parameters
in unit helix (``step'') direction
and one interpolation in the skew ``leap'' direction.
This interpolation scheme is also time-invariant
in the sense of item~\ref{item:time-invariance} in \secref{properties}
and \dfnref{time-invariance}
when we restrict the translation distances to multiples of the sampling period.

The proposed scheme is shown in \figref{grid-skew}.
\begin{figure}
\includegraphics[width=\columnwidth]{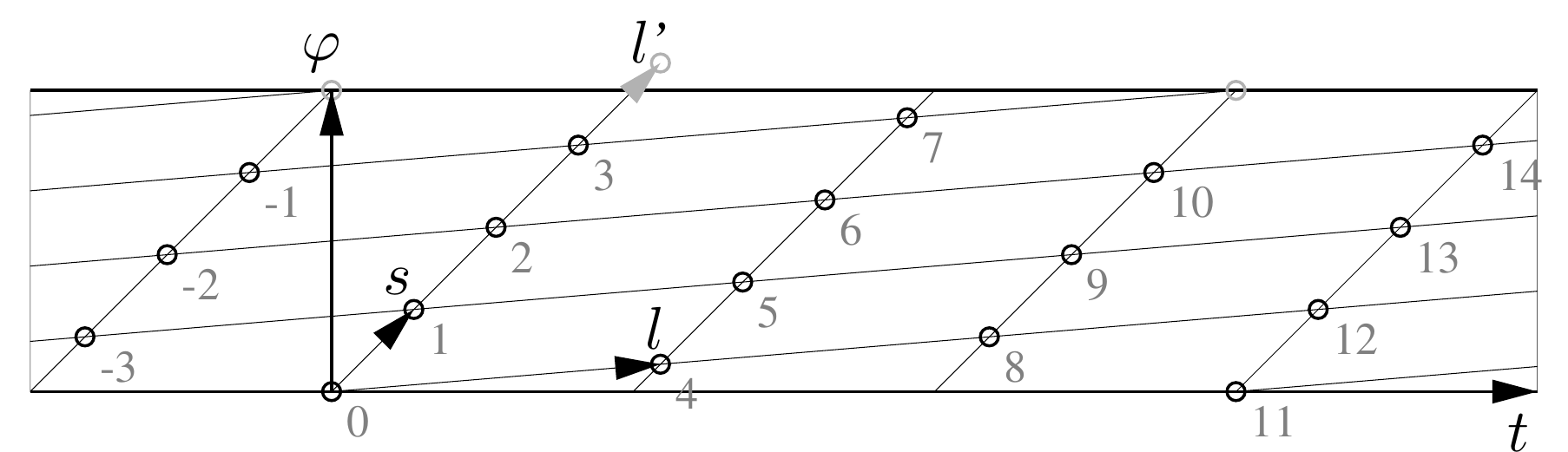}
\figcaption{grid-skew}{
Mapping of the sampled values to the cylinder in our method.
The variables $s$ and $l$ are coordinates in the skew coordinate system.
}
\end{figure}
We have a skew coordinate system with steps $s$ and leaps $l$.
We see, that this scheme can cope with non-integral wave periods,
that is, $T$ can be a fraction
(in \figref{grid-skew} we have $T=\frac{11}{3}$).
Whenever the wave period is integral,
the leap direction coincides with the time direction.
The grid nicely matches the periodic nature of the phase.
The cyclic phase yields ambiguities,
e.g. a leap could also go to where $l'$ is placed,
since this denotes the same signal value.
We will later see, that this ambiguity is only temporary
and will vanish at the end \eqnref{cell-coordinates}.
Thus we use the unique representative $\fromperiodic(\varphi)$ of $\varphi$.
To get $(l,s)$ from $(t,\fromperiodic(\varphi))$
we have to convert the coordinate systems,
i.e. we have to solve the simultaneous linear equations
\begin{eqnarray*}
\frac{1}{T}\cdot
\begin{pmatrix}
\round{T} & 1 \\
\round{T}-T & 1
\end{pmatrix}
\cdot
\begin{pmatrix}
l \\ s
\end{pmatrix}
&=&
\begin{pmatrix}
t \\ \fromperiodic(\varphi)
\end{pmatrix}
\end{eqnarray*}
where $\round$ is any rounding function we like.
E.g. in \figref{grid-skew} it is $\round{T} = 4$.
Its solution is
\begin{IEEEeqnarray*}{r/C/l}
l &=& t - \fromperiodic(\varphi) \eqnlabel{leap-solution} \\
s &=& t\cdot T - l\cdot\round{T}
\quad.
\end{IEEEeqnarray*}

Using the interpolated input $x$
we may interpolate $y$ linearly 
\begin{IEEEeqnarray*}{r/C/l}
r &=& \floor{l}\cdot\round{T} + s \\
\lerp(\xi,\eta)(\lambda) &=& \xi + \lambda \cdot (\eta-\xi)
   \eqnlabel{linear-interpolation} \\
\fractional\lambda &=& \lambda - \floor{\lambda} \\
y(t,\varphi)
   &=& \lerp\left(x(\tfrac{r}{T}), x(\tfrac{r+\round T}{T})\right)
            (\fractional{l})
\end{IEEEeqnarray*}
or more detailed
\begin{eqnarray*}
n &=& \floor{l}\cdot\round{T} + \floor{s} \\
a &=& \lerp(u(n),u(n+1))(\fractional{s}) \\
b &=& \lerp(u(n+\round{T}),u(n+\round{T}+1)) \\&&\qquad (\fractional{s}) \\
y(t,\varphi) &=& \lerp(a, b)(\fractional{l})
\quad.
\end{eqnarray*}
Actually, we do not even need to compute $s$
since by expansion of $s$ the formula for $r$ can be simplified
and it is $\fractional{s} = \fractional{r}$.
From $l$ we actually only need $\fractional{l}$.
This proves, that every representative of $\varphi$
could be used in $\eqnref{leap-solution}$.
\begin{IEEEeqnarray*}{r/C/l}
r &=& t\cdot T - \fractional{l}\cdot\round{T} \eqnlabel{cell-coordinates} \\
n &=& \floor{r} \\
a &=& \lerp(u(n),u(n+1))(\fractional{r}) \\
b &=& \lerp(u(n+\round{T}),u(n+\round{T}+1))(\fractional{r})
.
\end{IEEEeqnarray*}

\subsection{General Interpolations}
\seclabel{general-interpolations}

Other interpolations than the linear one
use the same computations to get $\fractional{l}$ and $r$,
but they access more values in the environment of $n$,
i.e. $u(n+j+k\cdot\round{T})$ for some $j$ and $k$.
E.g. for linear interpolation in the step direction
and cubic interpolation in the leap direction,
it is $j\in\{0,1\}, k\in\{-1,0,1,2\}$.

\subsection{Coping with Boundaries}
\seclabel{boundaries}

So far we have considered only signals that are infinite in both time directions.
When switching to signals with finite time domain
we become aware that our method consumes
more data than it produces at the boundaries.
This is however true for all interpolation methods.

We start considering linear interpolation:
In order to have a value for any phase at a given time,
a complete vertical bar must be covered by interpolation cells.
That happens the first time at time point $1$.
The same consideration is true for the end of the signal.
That is, our method always reduces the signal by two waves.
Analogously, for $k$ node interpolation in leap direction
we lose $k$ waves by pitch shifting.

If we would use extrapolation at the boundaries,
then for the same time but different phases
we would sometimes have to interpolate and sometimes we would extrapolate.
In order to avoid this, we just alter any $t\in[0,1)$ to $t = 1$
and limit $t$ accordingly at the end of the signal.

\subsection{Efficiency}
\seclabel{efficiency}

The algorithm for interpolating a value on the cylinder is actually very efficient.
The computation of the interpolation parameters
and signal value indices in \eqnref{cell-coordinates}
needs constant time,
and the interpolation is proportional
to the number of nodes in step direction
and the number of nodes in leap direction.
Thus for a given interpolation type,
generating an audio signal from the cylinder model
needs time proportional to the signal length
and only constant memory additional to the signal storage.

\subsection{Implementation}
\seclabel{implementation}

A reference implementation of the developed algorithm
is written in the purely functional programming language \haskell{}
\cite{peyton-jones1998haskell}.
The tree of modules is located at
\url{http://darcs.haskell.org/synthesizer/src/}.
In \cite{thielemann2004haskellsignal} we have already shown,
how this language fulfils the needs of signal processing.
The absence of side effects
makes functional programming perfect for parallelisation.
Recent progress on parallelisation in \haskell{}
\cite{peyton-jones2008dataparallel}
and the now wide availability of multi-core machines in the consumer market
justifies this choice.

We can generate the cylindrical wave function
with the function \code{Synthesizer.Basic.Wave.sampledTone}
given the interpolation in leap direction,
the interpolation in step direction,
the wave period of the input signal and the input signal.
The result of this function can then be used as input
for an oscillator that supports parametrised waveforms,
like \code{Synthesizer.Plain.Oscillator.shapeMod}.
By the way, this implementation again shows,
how functional programming with higher order functions supports modularisation:
The shape modulating oscillator can be used for any other kind
of parametrised waveform,
e.g. waveforms given by analytical functions.
This way, we have actually rendered the tones with morphing shape
in the figures of this paper.
In an imperative language you would certainly call the waveform
being implemented as call-back function.
However due to aggressive inlining
the compiled program does not actually need to callback the waveform function
but the whole oscillator process is expanded to a single loop.

\subsection{Streaming}
\seclabel{streaming}

Due to its lazy nature,
\haskell{} allows simple implementation of streaming,
that is, data is processed as it comes in,
and thus processing consumes only a constant amount of memory.
If we apply our pitch shifting and time stretching algorithm
to an ascending sequence of time values, streaming is possible.
This applies, since it is warranted,
that $\frac{r}{T}$ is not too far away from $t$.
Since $\fractional{l}\in[0,1)$ it holds
\begin{IEEEeqnarray*}{r/C/l}
t-\frac{r}{T} &\in& \left[0,\frac{\round{T}}{T}\right)
\eqnlabel{time-lag}
\quad.
\end{IEEEeqnarray*}
Thus we can safely move our focus to $t\cdot T - \round{T}$
in the discrete input signal $u$,
which is equivalent to a combined translation and turning
of the wave function on the cylinder.

What makes the implementation complicated is the handling of boundaries.
At the beginning we limit the time parameter as described in \secref{boundaries}.
However at the end, we have to make sure
that there is enough data for interpolation.
It is not so simple to limit $t$
to the length of input signal minus size of data needed for interpolation,
since determining the length of the input signal means reading it until the end.
Instead when moving the focus,
we only move as far as there is enough data available for interpolation.
The function is implemented by
\code{Synthesizer.Plain.\linebreak[3]Oscillator.shapeFreqModFromSampledTone}.

\section{Applications}
\seclabel{application}

\subsection{Combined pitch shifting and time scaling}
With a frequency control curve $f$ and a shape control $g$
we get combined pitch shifting and time scaling
out of our model using the conversion $S_{\int\hspace{-0.3em}f,\ g}$
(see \eqnref{cylinder-to-audio}).

\subsection{Wavetable synthesis}
Our algorithm might be used as alternative to wavetable synthesis
in sampling synthesisers \cite{massie1998wavetable}.
For wavetable synthesis a monophonic sound is reduced to a set of waveforms,
that is stored in the synthesiser.
On replay the synthesiser plays those waveforms successively in small loops,
maybe fading from one waveform to the next one.
If we do not reduce the set of waveforms,
but just chop the input signal into wave periods,
then apply wavetable synthesis with fading between waveforms,
we have something very similar to our method.
In \figref{trapezoid} we compare wavetable synthesis
and our algorithm using the introductory example of \figref{goal}.
In this example both the wavetable synthesis
and our method perform equally well.
If not stated otherwise,
in this and all other figures we use linear interpolation.
This minimises artifacts from boundary handling
and the results are good enough.
\begin{figure}
\begin{tabular}{l}
\includegraphics[width=\columnwidth]{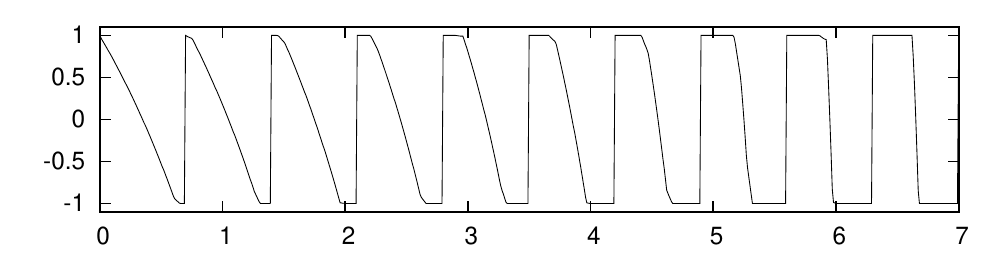} \\
\includegraphics[width=\columnwidth]{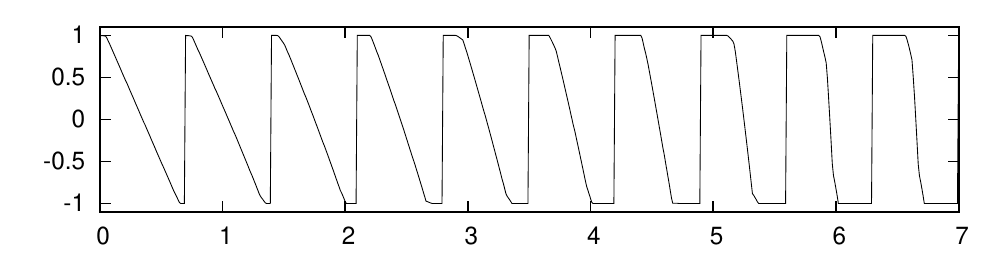}
\end{tabular}
\figcaption{trapezoid}{
Pitch shifting performed on the signal of \figref{goal}
using linear interpolation in both directions.
Above is the result of wavetable synthesis,
below is the result of our method.
}
\end{figure}

\subsection{Compression}
Wavetable synthesis can be viewed as a compression scheme:
Sounds are saved in the compressed form of a few waves
in the wavetable synthesiser
and are decompressed in realtime when playing the sound.
Analogously we can employ our method
for compression of monophonic sounds.
For compression we simply shrink the time scale
and for decompression we stretch it by the reciprocal factor.
An example is given in \figref{compression}.
\begin{figure}
\begin{tabular}{rl}
   & \vcentergraphics{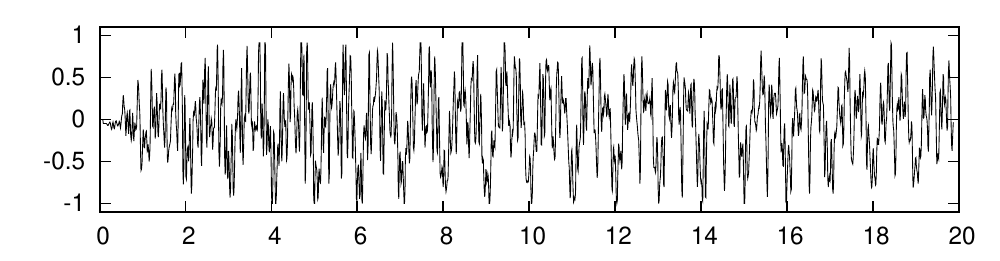} \\
 2 & \vcentergraphics{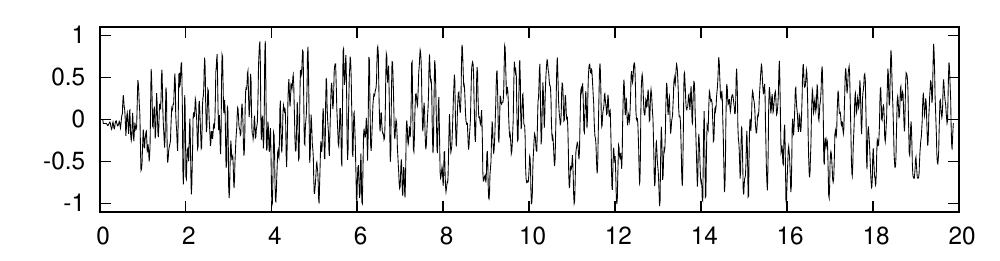} \\
 5 & \vcentergraphics{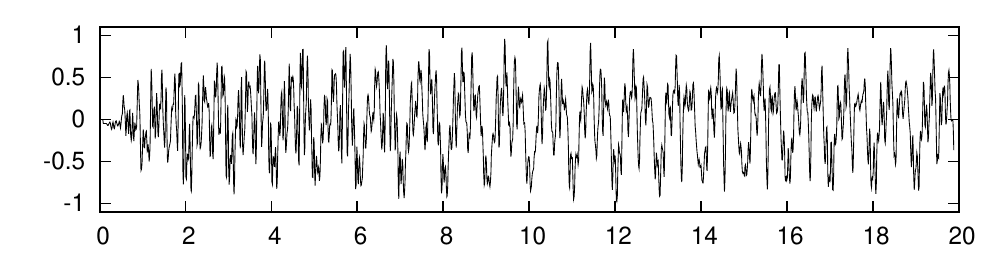} \\
10 & \vcentergraphics{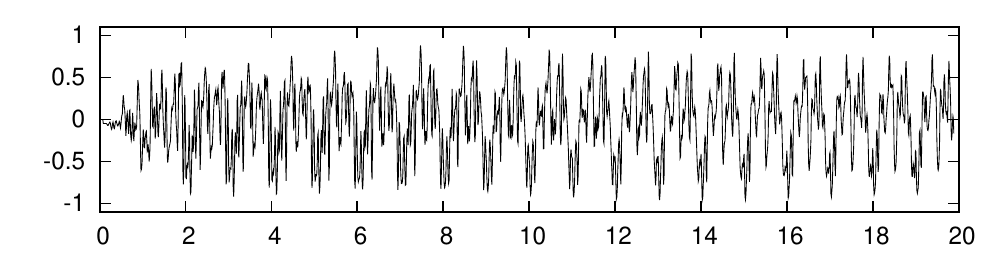} \\
25 & \vcentergraphics{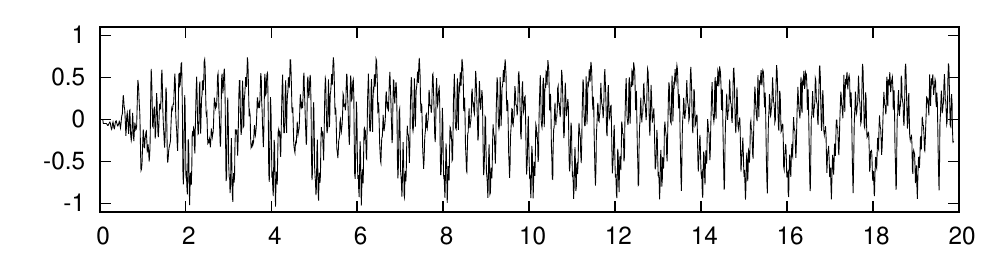} \\
50 & \vcentergraphics{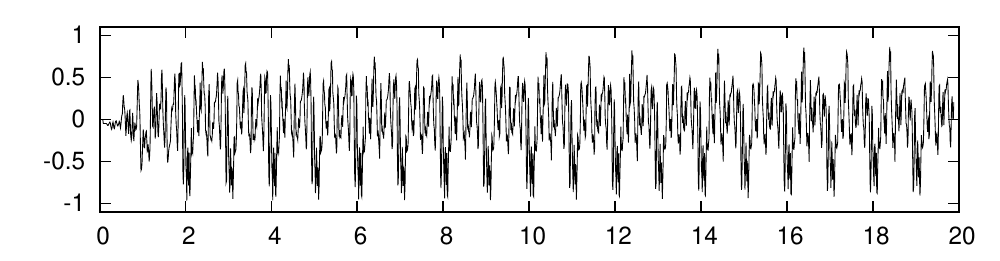}
\end{tabular}
\figcaption{compression}{
We show how a piano sound
is altered by compression and decompression.
The top-most graph is the original sound.
The graphs below are the results
of compression and decompression with cubic interpolation
by the associated factors in the left column.
Because the interpolation needs a margin at beginning,
we have copied the first two periods when compressing and decompressing.
}
\end{figure}

The shrinking factor, and thus the compression factor,
is limited by non-harmonic frequencies.
These are always present
in order to generate envelopes or phasing effects.
Consider the frequency~$a$
that is decomposed into $b+n$ as in \eqnref{decompose-frequency},
no pitch shift, i.e. $\alpha=1$,
and the shrinking factor~$v$.
According to \eqnref{frequency-mapping},
the frequency $b+n$ is mapped to $b\cdot v + n$.
In order to be able to decompose $b\cdot v + n$
into $b\cdot v$ and $n$ again on decompression,
it must be $b\cdot v\in(-\frac{1}{2},\frac{1}{2})$.
This implies, that if $b$ is the maximum absolute deviation
from an integral frequency, that you want to be able to reconstruct,
then it must be $v < \frac{1}{2\cdot b}$.

The mapping of frequencies can be best visualised
using the frequency spectrum
as in \figref{spectrum-compression}.
Note how the peaks become wider by the compression factor
while their shape is maintained.
The resolution is divided by the compression factor,
and this is why the compressed data actually consumes less space.
The shape of a peak expresses the envelope of the according harmonic
and widening it, means a time shrunken envelope.
\begin{figure}
\begin{tabular}{l}
\includegraphics[width=\columnwidth]{spectrum} \\
\includegraphics[width=\columnwidth]{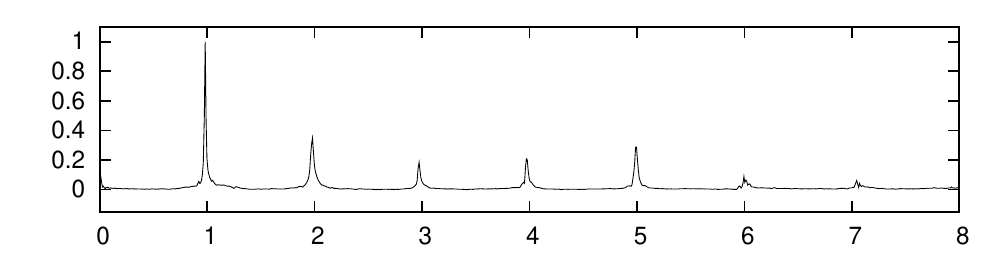}
\end{tabular}
\figcaption{spectrum-compression}{
The first graph presents the lower part
of the absolute spectrum of a piano sound.
This is then compressed by a factor~4 in the second graph.
}
\end{figure}

If we compress too much, then peaks will overlap
and we get aliasing effects on decompression.
Aliasing can be suppressed by smoothing
across the same phase of all waves.
That is, for the monophonic sound $x$ with period $T$
and a smoothing filter window $w$,
we should compress $x * (\upsample{\round{T}}{w})$ instead of $x$.
We use the up arrow for the upsampling operator where
\[
\forall \mset{k,c}\subset\Z \quad
 {(\upsample{c}{w})}_k =
  \begin{cases}
    w_{k/c} &: k\congruent 0 \mod c \\
    0 &: k\not\congruent 0 \mod c
  \end{cases}.
\]

Actually, we could use the frequency spectrum
not only for visualising the compression (or pitch-shifting),
but we could also use the frequency spectrum itself for compression.
The advantages would be simpler anti-aliasing
(we would just throw away values outside bands around the harmonics)
and we could also strip high harmonics,
once they fall below a given threshold.
The advantage of computing in the time-domain is,
that it consumes only linear time with respect to the signal length,
not linear-logarithmic time like the \Fourier{} transform,
that it can be applied in a streaming way
and allows to adapt the compression factor
to local characteristics of a sound.
For instance, you may use a shrinking factor close to 1
for fast varying portions of the signal
and use a larger shrinking factor on slowly modulated portions.

\subsection{Loop sampled sounds}
Another way to save memory in sampling synthesisers is to loop sounds.
This is especially important in order to get infinite sounds
like string sounds
out of a finite storage.
Looping means to repeat portions of a sampled sound.
The problem is to find positions of matching sound characteristics:
A loop that causes a jump or an abrupt change of the waveform
is a nasty audible artifact.
Especially in samples of natural sounds
there might be no such matching positions, at all.
Then the question is, whether the sample can be modified
in a way that preserves the sound but provides fine loop boundaries.
Several solutions using fading or time reversal have been proposed.

Our method offers a new way:
We may move the time forth and back while keeping pitch constant.
In \figref{loop-control} we show two reasonable time control curves.
Both control curves start with exactly reproducing the sampled sound
and then smoothly enter a cycle.
Actually, we copy the first part verbatim
instead of running time stretching with factor~1,
since our method cannot reproduce the beginning of the sound
due to interpolation margins.
The cycle of the first control curve consists of a sine,
that warrants smooth changes of the time line.
However with this control,
interferences are prolonged at the loop boundaries,
which is clearly audible.
It turns out that the second control curve,
namely the zig-zag curve, sounds better.
It preserves any chorus effect
and the change of the time direction is not as bad as expected.

A nice property of this approach is,
that the loop duration is doubled with respect to the actually looped data.
In contrast to that,
a loop body generated by simple cross-fading of parts of the sound,
say, with a \person{von Hann} window,
would half the loop body size and sounds more hectically.

Since the time control affects only the waveform,
it is warranted that at the cycle boundaries of the time control
the waveforms of the time manipulated sound match, too.
In order to assert the also the phases match
you have to choose a time control cycle length
that is an integral multiple of the wave period.
\begin{figure}
\begin{tabular}{l}
\includegraphics[width=\columnwidth]{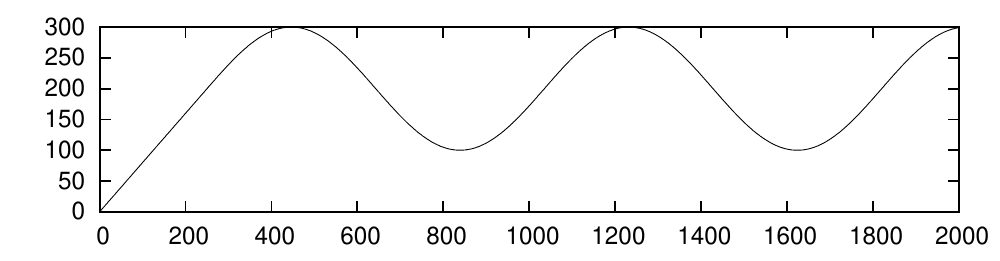} \\
\includegraphics[width=\columnwidth]{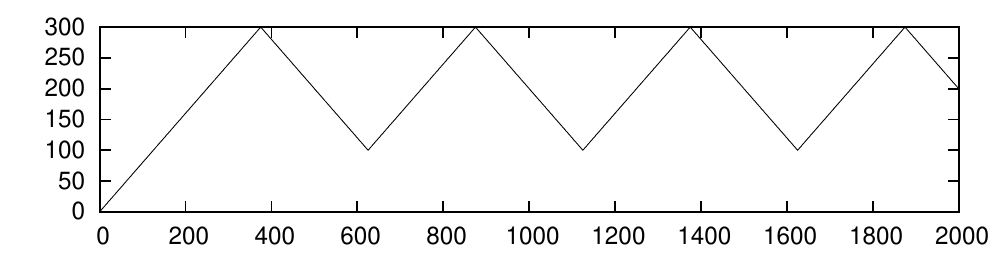} \\
\end{tabular}
\figcaption{loop-control}{
Two possible time control curves for generating a loopable portion
of a sampled sound.
}
\end{figure}

\subsection{Making inaudible harmonics audible}
Remember, that our model does not preserve formants.
Another application, where this is appropriate,
is to process sounds, where formants are not audible anyway,
namely ultrasound signals.
Our method can be used, to make monophonic ultrasound signals audible
by decreasing the pitch and while maintaining the length.
In \figref{bat} we show an echolocation call of a bat.
It is a chirp from about 35~kHz to 25~kHz sampled at 441~kHz.
The chirp nature does not match the requirements of our algorithm,
so it is not easy to choose a base frequency.
We have chosen 25~kHz and divide the frequency by factor~5
while maintaining the length.
Unfortunately the waves have no special form that we can preserve.
So this example might serve a demonstration
of the robustness of our algorithm with respect to non-harmonic frequencies
and the preservation of the envelope.
\begin{figure}
\begin{tabular}{l}
\includegraphics[width=\columnwidth]{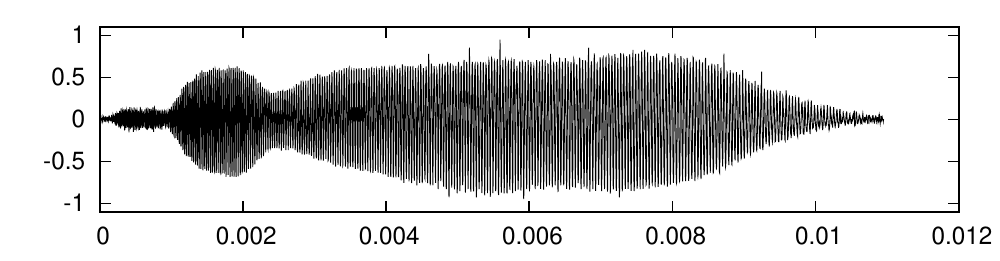} \\
\includegraphics[width=\columnwidth]{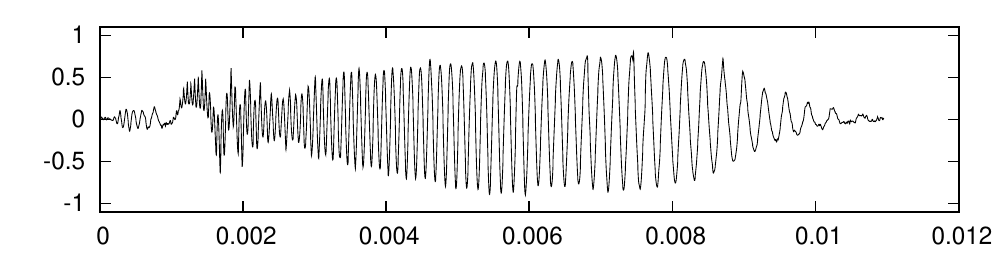}
\end{tabular}
\figcaption{bat}{
Echolocation call of Nyctalus noctula.
The time values are seconds.
}
\end{figure}
In the same way our method might be used to increase the pitch of infrasound.

\subsection{FM synthesis}
Since we can choose the phase parameter per sample,
we can not only do regular pitch shifting,
but we can also apply FM synthesis effects \cite{clowning1973fmsynthesis}.
An FM effect alone could also be achieved with synchronised time warping,
however with our method we can perform pitch shifting, time scaling
and FM synthesis in one go.
See \figref{fm-synthesis} for an example.
\begin{figure}
\begin{tabular}{l}
\includegraphics[width=\columnwidth]{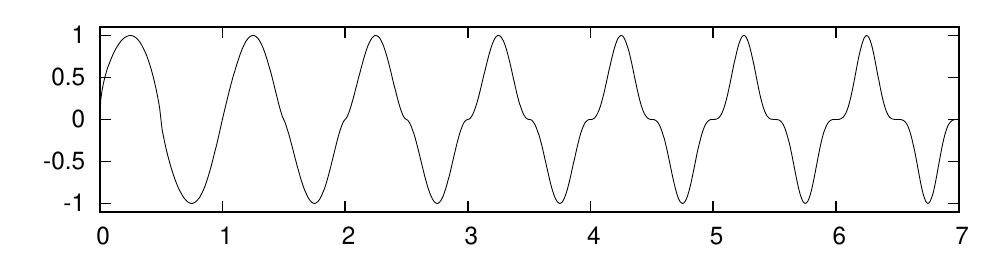} \\
\includegraphics[width=\columnwidth]{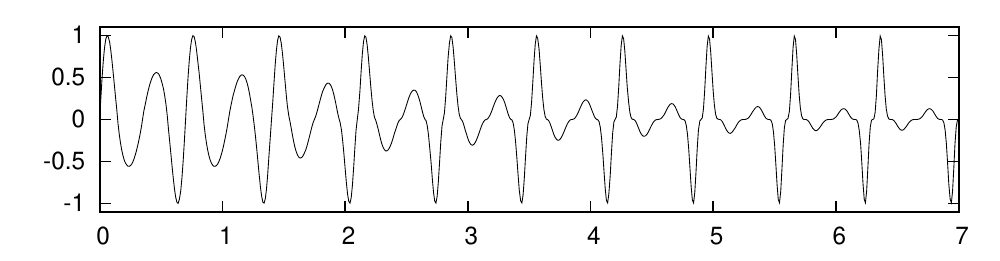}
\end{tabular}
\figcaption{fm-synthesis}{
Above is a sine wave that is distorted
by $\anonymfunc{v}{\sgnid v\cdot\abs{v}^p}$
for $p$ running from $\frac{1}{2}$ to $4$.
Below we applied our pitch shifting algorithm
in order to increase the pitch and
change the waveshape by modulating the phase
with a sine wave of the target frequency.
}
\end{figure}

\subsection{Tone generation by time stretching}
The inability to reproduce noise
can be used for creative effects.
By time stretching we can get a tone out of every sound.
This is exemplified in \figref{noisy-tone}.
If we stretch time by a factor~$n$
for a specific period~$T$
(source and target period shall be equal),
then in the spectrum the peak for each harmonic of frequency~$\frac{1}{T}$
is narrowed by a factor~$n$.

\begin{figure}
\begin{tabular}{l}
\includegraphics[width=\columnwidth]{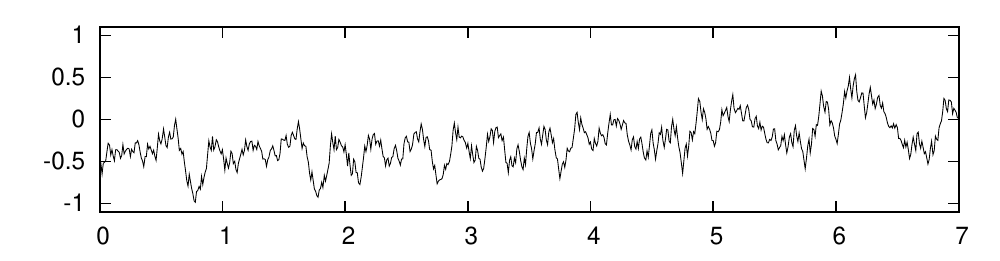}
\end{tabular}
\figcaption{noisy-tone}{
A tone generated from pink noise by time stretching.
The source and the target period are equal.
The time is stretched by factor~4.
}
\end{figure}

\section{Related work}
\seclabel{related-work}


The idea of separating parameters (here phase and shape)
that are in principle indistinguishable is not new.
For example it is used in \cite{lang2002network}
for separation of sine waves of considerably different frequencies.
This way a numerically problematic ordinary differential equation
is turned into a well-behaved partial differential equation.

Also the specific tasks of pitch shifting and time scaling
are addressed by a broad range of algorithms \cite{zoelzer2002dafx}.
Some of them are intended for application on complex music signals
and are relatively simple,
like ``Overlap and Add'' (OLA), ``Synchronous Overlap and Add'' (SOLA)
\cite{roucos1985timescale,makhoul1986timescale},
or the three-phase overlap algorithm using cosine windows
presented in \cite{disch1999modulation}.
They take segments of an audio signal as they are,
rearrange them and reduce the artifacts of the new composition.
Other methods are based on a model of the sound.
E.g. ``pitch-synchronous overlap-add'' (PSOLA)
is roughly based on the excitation+filter model for speech
\cite{hamon1989diphone,moulines1990pitchsync,lemmetty1999speechsynthesis},
sinusoidal models interpret sounds as mixture of sine waves
that are mo\-dulated in amplitude and frequency \cite{raspaud2007resampling},
even more sophisticated models
treat sounds as mix of sine waves, transients and a residual
\cite{nsabimana2008audiodecomposition}.
There are also methods specific to monophonic signals,
like wavetable synthesis \cite{massie1998wavetable}
and advanced methods, that can cope with frequency modulated input signals
\cite{haghparast2007pitchshifting}.

In the following two sections we like to compare our method
with the two methods that are most similar to the one we introduced here,
namely with wavetable synthesis and PSOLA.

\subsection{Comparison with Wavetable Synthesis}
\seclabel{comparison-wavetable}

When we chop our input signal into wave periods
and use the waves as wavetable,
then wavetable synthesis becomes rather similar to our method
\cite{massie1998wavetable}.
Wavetable synthesis also preserves waveforms, rather than formants,
it allows frequency and shape modulation at sample rate.
However, due to the treatment of waveforms as discrete objects,
the wavetable synthesis cannot cope well with non-harmonic frequencies
(\figref{sine}).
Thus, in wavetable synthesisers, phasing is usually implemented
using multiple wavetable oscillators.
A minor deficiency is,
that fractional periods of the input signal are not supported.
The wavetables always have to have an integral length.
We consider this deficiency to be not so important,
since when we do not match the wave period exactly,
this will appear to the wavetable synthesis algorithm as a shifting waveform.
But that algorithm must handle varying waveshapes anyway.

The wavetables in a wavetable synthesiser
are usually created by a more sophisticated preprocessing
than just chopping a signal into pieces of equal length.
However, for comparison purposes we will just use this simple procedure.

Chopping and subsequent wavetable synthesis can also be interpreted
as placing the sample values on a cylinder
and interpolating between them.
It yields the pattern shown in \figref{grid-ortho}.
\begin{figure}
\includegraphics[width=\columnwidth]{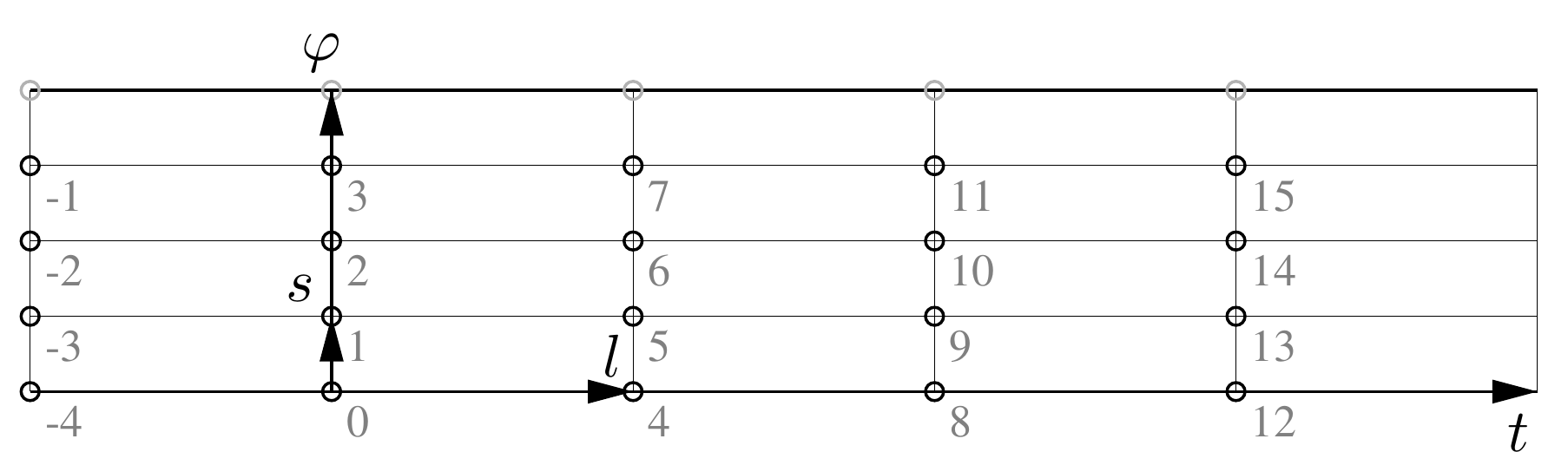}
\figcaption{grid-ortho}{
Mapping of the sampled values to the cylinder
in the wavetable-oscillator method.
The grey numbers are the time points in the input signal.
}
\end{figure}
The variable $s$ denotes the ``step'' direction,
which coincides with the direction of the phase in this scheme.
The variable $l$ denotes the ``leap'' direction,
which coincides with the time direction.
In order to fit the requirement of a wave period of 1
we shrink the discrete input signal.
Say, the discrete input signal is $u$,
the wave period is $T$, that must be integral,
and the real input signal is $x$,
that we define at some discrete fractional points
by $x(\frac{n}{T}) = u(n)$
and at the other ones by interpolation.
In \figref{grid-ortho} it is $T=4$
and for example $y(1.7,\toperiodic(0.6))$
is located in the rectangle spanned
by the time points $6, 7, 10, 11$.
For simplicity let us use linear interpolation as in
\eqnref{linear-interpolation}.
We would interpolate
\begin{multline*}
y(1.7)(\toperiodic(0.6)) = \\
 \lerp(\lerp(u(6),u(7))(0.4), \lerp(u(10),u(11))(0.4))(0.7)
.
\end{multline*}
In general for $y(t,\varphi)$ we get
\begin{eqnarray*}
\forall r\in\R\quad
   \fractional{r} &=& r - \floor{r} \\
\forall r\in\R\quad
   x(\tfrac{r}{T}) &=& \lerp(u(\floor{r}), u(\floor{r}+1))(\fractional{r}) \\
\tau &=& \floor{t} + \fromperiodic(\varphi) \\
y(t,\varphi)
   &=& \lerp(x(\tau), x(\tau + 1))(\fractional{t})
\end{eqnarray*}
or more detailed
\begin{eqnarray*}
s &=& T\cdot\fromperiodic(\varphi) \\
n &=& T\cdot\floor{t} + \floor{s} \\
a &=& \lerp(u(n),u(n+1))(\fractional{s}) \\
b &=& \lerp(u(n+T),u(n+T+1))(\fractional{s})) \\
y(t,\varphi) &=& \lerp(a, b)(\fractional{t})
.
\end{eqnarray*}
The handling of waveform boundaries points us to a problem of this method:
Also at the waveform boundaries we interpolate bet\-ween adjacent values
of the input signal $u$.
That is, we do not wrap around.
This way, waveforms can become discontinuous by interpolation.
We could as well wrap around the indices at waveform boundaries.
This would complicate the computation and raises the question,
what values should naturally be considered neighbours.
We remember,
that we also have the ambiguity of phase values in our method.
But there, the ambiguity vanishes in a subsequent step.

\subsubsection{Boundaries}

If we have an input signal of $n$ wave periods,
then we have only $n-1$ sections where we can interpolate linearly.
Letting alone that this approach cannot reconstruct a given signal,
it loses one wave at the end for linear interpolation.
If there is no integral number of waves,
than we may lose up to (but excluding) two waves.
For interpolation between $k$ nodes in time direction
we lose $k-1$ waves.
Of course, we could extrapolate, but this is generally problematic.

That is, the wavetable oscillator cuts away between one and two waves,
whereas our method always reduces the signal by two waves.
Thus the wavetable oscillator is slightly more economic.

\subsection{Comparison with PSOLA}
\seclabel{comparison-psola}

Especially for speech processing,
we would have to preserve formants rather than waveshapes.
The standard method for this application is
``(Time Domain) Pitch-Synchronous Overlap/Add'' (TD-PSOLA)
\cite{hamon1989diphone,moulines1990pitchsync}.
PSOLA decomposes a signal into wave atoms,
that are rearranged and mixed
while maintaining their time scale.
The modulation of the timbre and the pitch
can only be done at wave rate.
As for wavetable synthesis it is also true for PSOLA,
that due to the discrete handling of waveforms,
non-harmonic frequencies are not handled well.


Incidentally, time shrinking at constant pitch with our method
is similar to PSOLA of a monophonic sound.
For time shrinking with factor~$v$
and interpolating with kernel $\kappa$
our algorithm computes:
\begin{eqnarray*}
z(t)
   &=& y(v\cdot t,\toperiodic(t)) \\
   &=& \sum_{k\in\Z} x(t+k)\cdot\kappa(v\cdot t-(t+k)) \\
   &=& \sum_{k\in\Z} x(t+k)\cdot\kappa((v-1)\cdot t-k) \\
\eqntext{with $(\shrink{d}{\kappa})(t) = \kappa(d\cdot t)$}
z &=& \sum_{k\in\Z}
           (\translatel{k}{x})\cdot
           (\shrink{(v-1)}{(\translater{k}{\kappa})})
\quad.
\end{eqnarray*}
We see that the interpolation kernel $\kappa$
acts like the segment window in PSOLA,
but it is applied to different phases of the waves.
For $v=1$, only the non-translated $x$ is passed to the output.

Intuitively we can say,
that PSOLA is source oriented or push-driven,
since it dissects the input signal into segments
independent from what kind of output is requested.
Then it computes, where to put these segments in the output.
In these terms, our method is target oriented or pull-driven,
as it investigates for every output value,
where it can get the data for its construction from.

Actually, it would be easy to add another parameter to PSOLA
for time stretching the atoms.
This way one could interpolate bet\-ween shape preservation and formant preservation.

\section{Results and comparisons}
\seclabel{example}

Finally we like to show some more results of our method
and compare them with the wavetable synthesis.

In \figref{enveloped} we show,
that signals with band-limited amplitude modulation
can be perfectly reconstructed, except at the boun\-daries.
Although we do not employ \person{Whittaker} interpolation
but simple linear interpolation the result is convincing.
\begin{figure}
\begin{tabular}{l}
\includegraphics[width=\columnwidth]{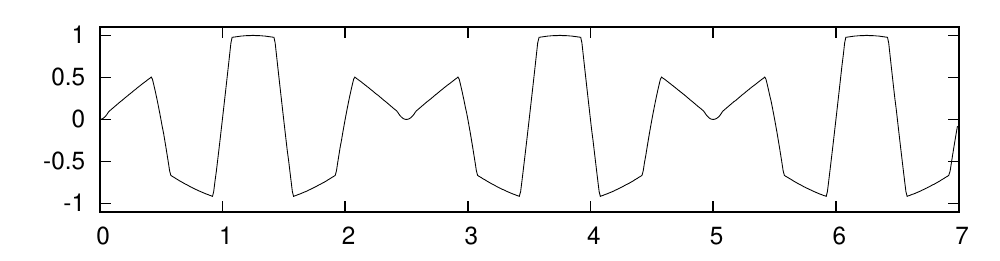} \\
\includegraphics[width=\columnwidth]{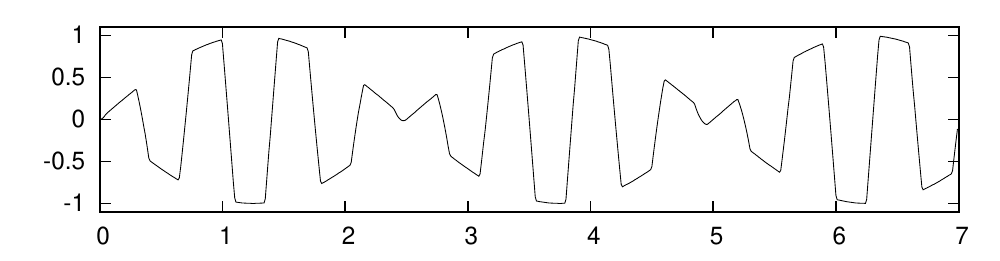} \\
\includegraphics[width=\columnwidth]{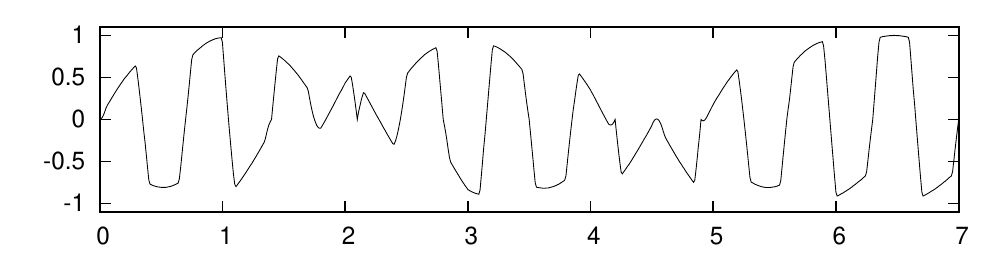} \\
\includegraphics[width=\columnwidth]{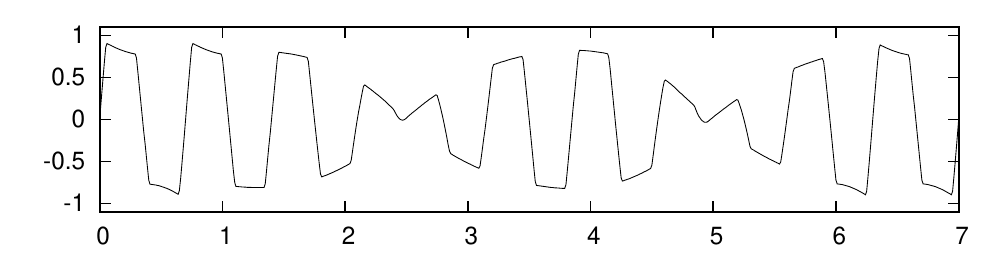}
\end{tabular}
\figcaption{enveloped}{
Pitch shifting performed on a periodically amplitude modulated tone
using linear interpolation.
The figures show from top to bottom:
The input signal,
the signal recomputed with a different pitch
(that is, the ideal result of a pitch shifter),
the result of wavetable oscillating,
the result of our method.
}
\end{figure}

In \figref{sine} we apply our method to a sine
with a frequency that is clearly distinct from $1$.
To a monophonic pitch shifter this looks like a rapidly changing waveform.
As derived for \person{Whittaker} interpolation in \eqnref{sine-preserved}
our method can at least reconstruct the sine shape,
however the frequency of the pitch shifted signal differs from the intended one.
Again, the used linear interpolation does not seem to be substantially worse.
\begin{figure}
\begin{tabular}{l}
\includegraphics[width=\columnwidth]{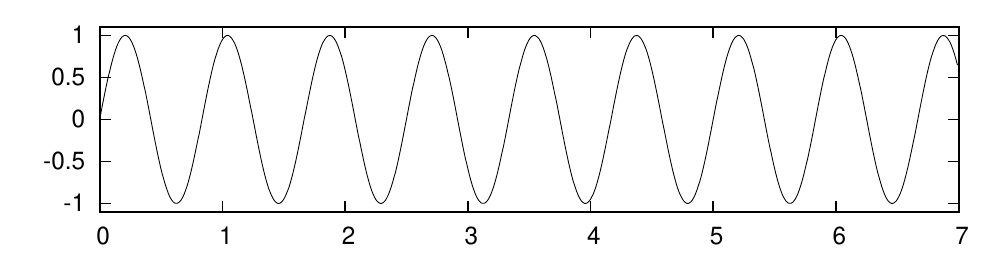} \\
\includegraphics[width=\columnwidth]{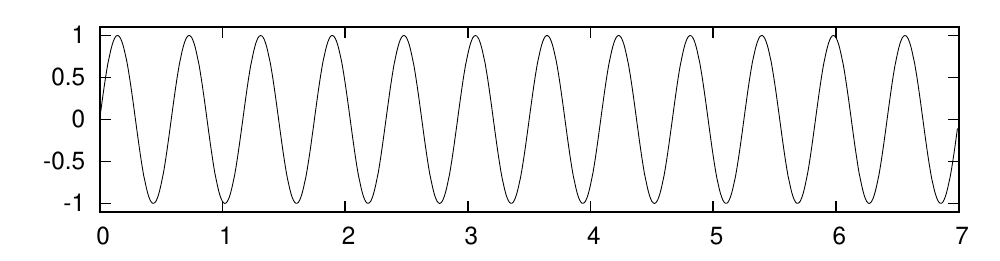} \\
\includegraphics[width=\columnwidth]{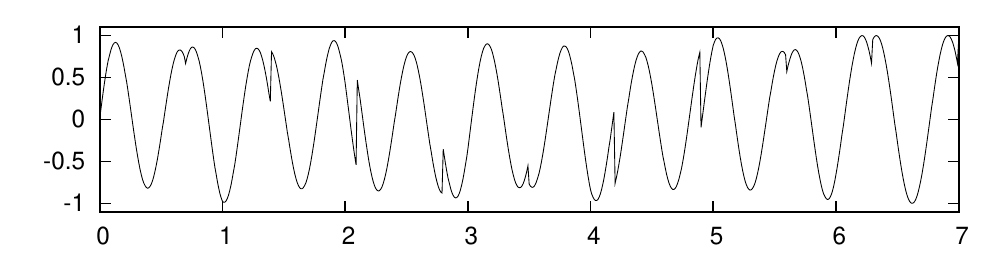} \\
\includegraphics[width=\columnwidth]{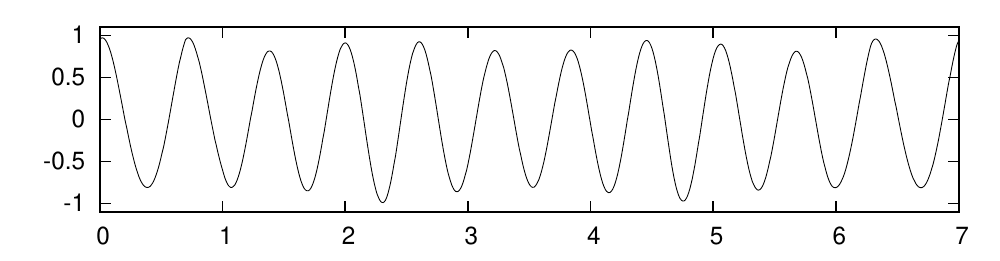}
\end{tabular}
\figcaption{sine}{
Pitch shifting performed on a sine tone with a frequency
that deviates from the required frequency $1$.
The graphs are arranged analogously to \figref{enveloped}.
}
\end{figure}

We also like to show how phase modulation at sample rate
can be used for FM synthesis combined with pitch shifting.
In \figref{fm-synthesis-power} we use
a sine wave with changing distortion as input,
whereas in \figref{fm-synthesis-sined}
the sine wave is not distorted, but detuned to frequency $1.2$,
which must be treated as changing waveform with respect to frequency $1$.
\begin{figure}
\begin{tabular}{l}
\includegraphics[width=\columnwidth]{power07} \\
\includegraphics[width=\columnwidth]{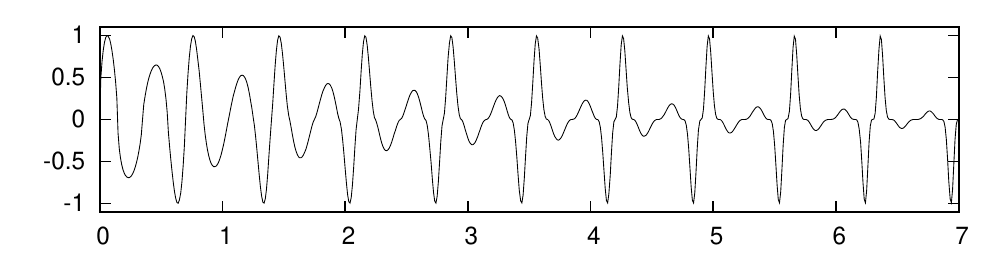} \\
\includegraphics[width=\columnwidth]{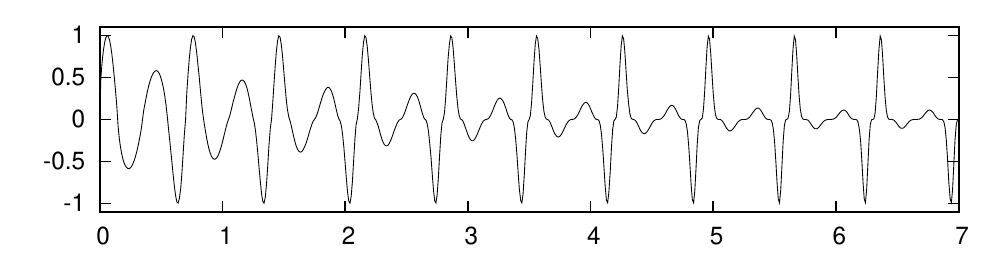} \\
\includegraphics[width=\columnwidth]{power10fm-from-helix}
\end{tabular}
\figcaption{fm-synthesis-power}{
Above is a sine wave that is distorted
by $\anonymfunc{v}{\sgnid v\cdot\abs{v}^p}$
for $p$ running from $\frac{1}{2}$ to $4$.
Below we applied our pitch shifting algorithm
in order to increase the pitch and
change the waveshape by modulating the phase
with a sine wave of the target frequency.
The graphs are arranged analogously to \figref{enveloped}.
}
\end{figure}
\begin{figure}
\begin{tabular}{l}
\includegraphics[width=\columnwidth]{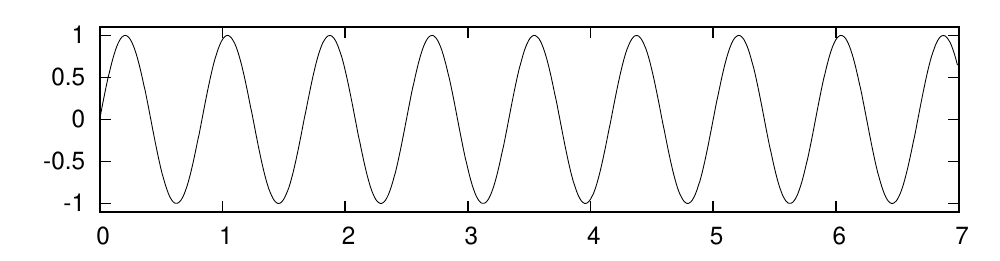} \\
\includegraphics[width=\columnwidth]{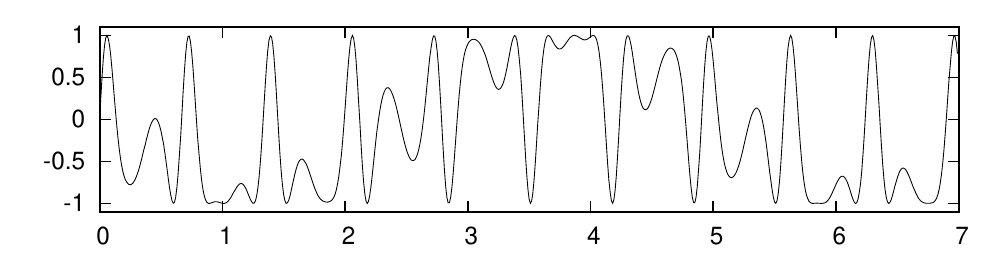} \\
\includegraphics[width=\columnwidth]{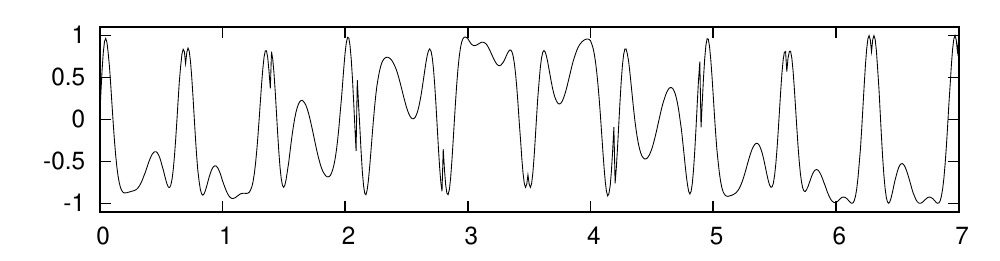} \\
\includegraphics[width=\columnwidth]{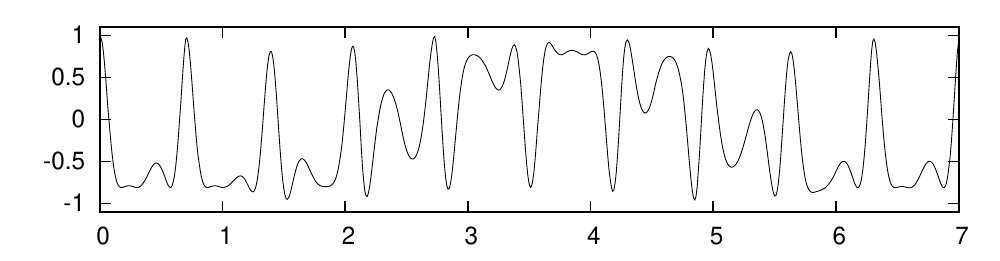}
\end{tabular}
\figcaption{fm-synthesis-sined}{
Here we demonstrate FM synthesis
where the carrier sine wave is detuned.
The graphs are arranged analogously to \figref{enveloped}.
}
\end{figure}

As a kind of counterexample we demonstrate in \figref{percussive},
how the boundary handling forces our method
to limit the time parameter to values above 1
and thus it cannot reproduce the beginning of the sound properly.
\begin{figure}
\begin{tabular}{l}
\includegraphics[width=\columnwidth]{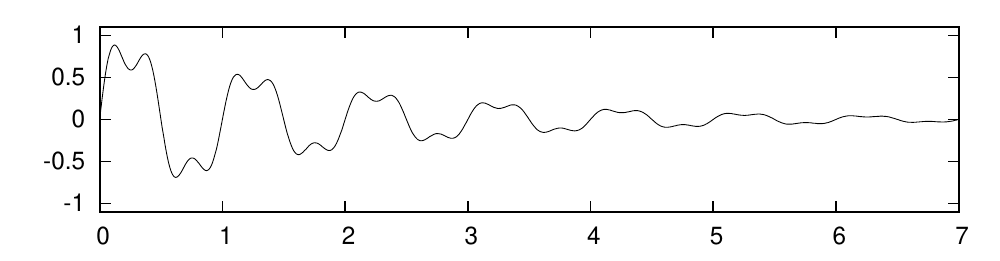} \\
\includegraphics[width=\columnwidth]{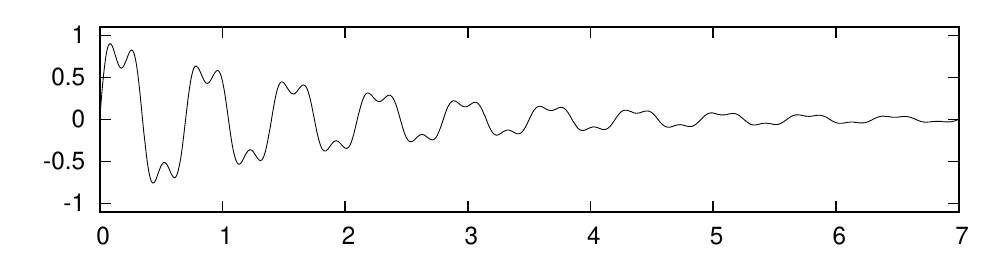} \\
\includegraphics[width=\columnwidth]{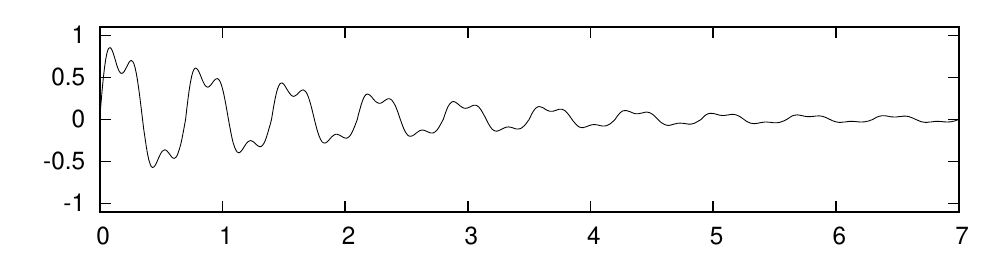} \\
\includegraphics[width=\columnwidth]{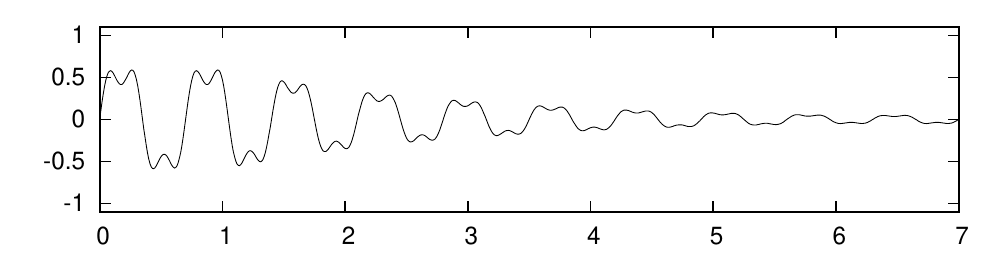}
\end{tabular}
\figcaption{percussive}{
Pitch shifting performed on a percussive tone.
The graphs are arranged analogously to \figref{enveloped}.
}
\end{figure}
For completeness we also present the same sound
transposed by PSOLA in \figref{percussive-psola}.
\begin{figure}
\includegraphics[width=\columnwidth]{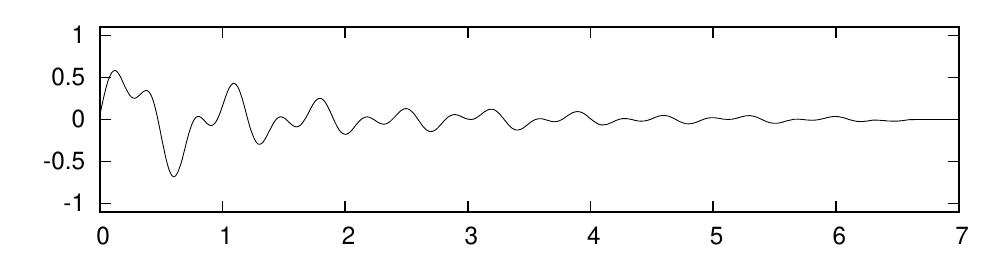} \\
\figcaption{percussive-psola}{
Pitch shifting with the tone from \figref{percussive}
that preserves formants performed by PSOLA.
}
\end{figure}

Please note that the examples have a small number of periods (7 to 10)
compared to signals of real instruments (say, 200 to 2000 per second).
On the one hand, graphs of real world sounds
would not fit on the pages of this journal at a reasonable resolution.
On the other hand, only for those small numbers of periods
we get a visible difference between the methods we compare here.
However, if you are going to implement a single tone pitch shifter from scratch
you might prefer our method, because it handles the corner cases better
and the complexity is comparable to that of the wavetable oscillator.
Also for theoretical considerations we recommend our method
since it exposes the nice properties presented in \secref{cont-signal}.

\subsection{Conclusions}

We shall note that despite the differences
between our method and existing ones,
many of the properties discussed in \secref{properties}
hold approximately also for the existing methods.
Thus the worth of our work is certainly
to contribute a model where these properties apply exactly.
This should serve a good foundation
for further development of a sound theory
of pitch shifting and time scaling.
It also pays off, when it comes to corner cases,
like FM synthesis as extreme pitch shifting.

\section{Outlook}
\seclabel{outlook}

\subsection{Band Limitation}
In our paper we have omitted how to avoid aliasing effects
in pitch shifting caused by too high harmonics in the waveforms.
In some way we have to band-limit the waveforms.
Again, we should do this without actually constructing
the two-dimensional cylindrical function.
When we use interpolation that does not extend the frequency band,
that is imposed by the discrete input signal,
then it should be fine to lowpass filter the input signal
before converting to the cylinder.
The cut-off frequency must be dynamically adapted to the frequency modulation
used on conversion from the cylinder to the audio signal.

\subsection{Irregular Interpolation}
We could also handle input of varying pitch.
We would then need a function of time describing the frequency modulation
which is used to place the signal nodes at the cylinder.
This would be an irregular pattern
and renders the whole theory of \secref{cont-signal-theory} useless.
We had to choose a generalised 2D interpolation scheme.

\section{Acknowledgments}

I like to thank Alexander Hinneburg for
fruitful discussions and creative suggestions.
I also like to acknowledge Sylvain Marchand and Martin Raspaud
for their comments on my idea and their encouragement.
Finally I am grateful to Stuart Parsons,
who kindly permitted usage of his bat recordings in this paper.

\bibliographystyle{IEEEbib}
\bibliography{audio,thielemann,haskell,literature,wavelet}

\appendix

\section{Automated proofs with PVS}
\seclabel{pvs-proofs}

The goal of proof assistants is currently not to simplify proving,
but to get confidence that a claim is true.
Actually, you will succeed with a proof
only with a profound understanding of the problem
and preferably several proof ideas,
of which only one can be enough formalised
such that the proof assistant accepts it.

To give an impression of automated proving,
we show the derivation of time-invariant interpolations
from \secref{time-invariance}
expressed by two lemmas in PVS \cite{owre2001pvs}
in \figref{pvs-time-invariant-ip}.
See \url{http://darcs.haskell.org/synthesizer/src/Synthesizer/Plain/ToneModulation/}
for the according modules.
\begin{figure}
\begin{verbatim}
Displacement: TYPE = real
Time: TYPE = real
Phase: TYPE =
  Quotient(LAMBDA (p0, p1):
             integer?(p1 - p0))
Signal: TYPE = [Time -> Displacement]
Waveform: TYPE = [Phase -> Displacement]
Tube: TYPE = [Time -> Waveform]

t: VAR Time
x: VAR Signal
F: VAR [Signal -> Tube]
I: VAR [Signal -> Waveform]

IS(I)(x)(t): Waveform =
  rotate_right(t)(I(translate_left(t)(x)))

time_invariant?(F): bool =
  FORALL x, t:
    F(translate_right(t)(x)) =
       translate2(t, t)(F(x))

interpolation_time_invariant: LEMMA
  time_invariant?(IS(I))

interpolation_slice: LEMMA
  time_invariant?(F) =>
    (EXISTS I: F = IS(I))
\end{verbatim}
\figcaption{pvs-time-invariant-ip}{
Excerpt from a PVS module containing two statements:
The first claim is that the interpolation of the form
given in \eqnref{time-invariant-ip} is time-invariant
in the sense of \dfnref{time-invariance}.
The second claim is that all time-invariant interpolations
can be expressed in that form.
In contrast to the PVS language,
the according proof script can only be understood
when interactively running it step by step in PVS
and looking at how the expressions evolve.
}
\end{figure}
The lemma, that constant interpolation preserves static waves
is shown in \figref{pvs-static-wave}.
See \secref{static-wave-preservation} for details.
\begin{figure}
\begin{verbatim}
w: VAR Waveform

constant_tube?(y): bool =
  FORALL t0, t1: y(t0) = y(t1)

interpolation_constant: LEMMA
  FORALL w: constant_tube?
    (IS(LAMBDA x: x o cinv)(w o c))
\end{verbatim}
\figcaption{pvs-static-wave}{
PVS lemma that claims
that the constant interpolation preserves static waves.
}
\end{figure}

\end{document}